\documentclass[preprint]{imsart}

\RequirePackage[OT1]{fontenc}
\RequirePackage{amsthm,amsmath}
\RequirePackage[round]{natbib}
\RequirePackage[colorlinks,citecolor=blue,urlcolor=blue]{hyperref}

\arxiv{arXiv:0000.0000}

\startlocaldefs
\usepackage{booktabs}
\usepackage{amsfonts,amssymb}
\usepackage{graphicx}
\usepackage{paralist}
\numberwithin{equation}{section}
\theoremstyle{plain}

\newtheorem{lemma}{Lemma}[section]
\newtheorem{theorem}{Theorem}[section]
\theoremstyle{definition}
\newtheorem{remark}{Remark}
\newcommand{\dnorm}{\mathcal{N}}

\newcommand{\cx}{\mathcal X}
\newcommand{\cy}{\mathcal Y}

\newcommand{\e}{\mathbb{E}}

\newcommand{\ssp}{\mathrm{S}}
\newcommand{\brp}{\mathrm{B}}
\newcommand{\wh}{\widehat}
\newcommand{\tran}{\mathsf{T}}
\newcommand{\err}{\varepsilon}

\newcommand{\diag}{\mathrm{diag}}
\newcommand{\df}{\mathrm{df}}
\newcommand{\tr}{\mathrm{tr}}

\newcommand{\wt}{\widetilde}

\newcommand{\var}{\mathrm{var}}
\newcommand{\cov}{\mathrm{cov}}
\newcommand{\mse}{\mathrm{MSE}}

\newcommand{\orcl}{\mathrm{orcl}}
\newcommand{\pool}{\mathrm{pool}}
\newcommand{\plug}{\mathrm{plug}}
\newcommand{\bapi}{\mathrm{bapi}}

\newcommand{\cv}{\mathrm{cv}}

\newcommand{\real}{\mathbb{R}}

\newcommand{\simiid}{\stackrel{\text{iid}}{\sim}}

\newcommand{\rd}{\,\mathrm{d}}

\newcommand{\sign}{\mathrm{sign}}

\newcommand{\onedc}{1_{|D|\ge c}}
\newcommand{\js}{\mathrm{JS}}

\endlocaldefs

\begin{document}

\begin{frontmatter}
\title{Data enriched linear regression}
\runtitle{Data enriched regression}

\begin{aug}
\author{\fnms{Aiyou} \snm{Chen}\thanksref{m1}\ead[label=e1]{aiyouchen@google.com}},
\author{\fnms{Art B.} \snm{Owen}\thanksref{m2}\thanksref{t1}\ead[label=e2]{owen@stanford.edu}}
\and
\author{\fnms{Minghui} \snm{Shi}\thanksref{m1}
\ead[label=e3]{mshi@google.com}
\ead[label=u1,url]{http://stat.stanford.edu/$\sim$owen}}

\thankstext{t1}{Art Owen worked on this project
as a consultant for Google; it was not part of
his Stanford responsibilities.}
\runauthor{A. Chen et al.}

\affiliation{Google Inc.\thanksmark{m1}, Stanford University\thanksmark{m2}}

\address{Google Inc.\\
1600 Amphitheatre Pkwy\\
Mountain View, CA 94303\\
\printead{e1,e3}}

\address{Department of Statistics\\
Sequoia Hall\\
Stanford, CA 94305\\
\printead{e2}\\
\printead{u1}}
\end{aug}

\begin{abstract}
We present a linear regression method for predictions
on a small data set making use of a second possibly
biased data set that may be much larger.  Our method
fits linear regressions to the two data sets while penalizing
the difference between predictions made by those two models.
The resulting algorithm is a shrinkage method similar
to those used in small area estimation. 
We find a Stein-type finding for Gaussian responses: 
when the model has $5$ or more coefficients and $10$
or more error degrees of freedom, it becomes inadmissible
to use only the small data set, no matter how large the
bias is.  
We also present both plug-in and AICc-based methods
to tune our penalty parameter. Most of our results use
an $L_2$ penalty, but we obtain formulas for $L_1$
penalized estimates when the model is specialized to the
location setting.
Ordinary Stein shrinkage provides an inadmissibility
result for only $3$ or more coefficients, but we find that
our shrinkage method typically produces much lower squared errors in
as few as $5$ or $10$ dimensions when the bias is small
and essentially equivalent squared errors when the bias is large.
\end{abstract}

\begin{keyword}[class=MSC]
\kwd[Primary ]{62J07}
\kwd{62D05}
\kwd[; secondary ]{62F12}
\end{keyword}

\begin{keyword}
\kwd{Data fusion}
\kwd{Stein shrinkage}
\kwd{Transfer learning}
\end{keyword}

\end{frontmatter}
\section{Introduction}
The problem we consider here is how to combine
linear regressions based on 
data from two sources.  There is a small data set
of expensive high quality observations and a possibly much larger
data set with less costly observations.
The big data set is thought to have similar
but not identical statistical characteristics
to the small one.  The conditional expectation might
be different there or the predictor variables might
have been measured in somewhat different ways.
The motivating application comes from within Google.
The small data set is a panel of
consumers, selected by a probability sample,
who are paid to share their internet viewing data
along with other data on television viewing.
There is a second and potentially much larger panel, not selected
by a probability sample who have  opted in
to the data collection process.

The goal is to make predictions for the population from which the smaller 
sample was drawn. If the data are identically
distributed in both samples, we should simply
pool them.  If the big data set is completely
different from the small one, then using it
may not be worth the trouble.

Many settings are intermediate
between these extremes: the big data set is similar
but not necessarily identical to the small one.
We stand to benefit from using the big data set
at the risk of introducing some bias.  
Our goal is to glean some information from
the larger data set to increase accuracy
for the smaller one.
The difficulty is that our best information about
how the two populations are similar is in our
samples from them.

The motivating problem at Google has some differences
from the problem we consider here.  
There were two binary responses, one sample was missing
one of those responses, and tree-based predictions were used.
See \cite{chen:koeh:owen:remy:shi:2014}.
This paper studies linear regression because it is more amenable to theoretical
analysis, is more fundamental, and sharper statements are possible.

The linear regression method we use is a hybrid
between simply pooling the two data sets and fitting
separate models to them. As explained in more
detail below, we apply
shrinkage methods penalizing the difference between
the regression coefficients for the two data sets.
Both the specific penalties we use, and our
tuning strategies, reflect our greater
interest in the small data set. 
Our goal is to enrich the analysis of the smaller
data set using possibly biased data from the larger
one.

Section~\ref{sec:nota} presents our notation
and introduces $L_1$ and $L_2$ penalties
on the parameter difference. Most of our results
are for the $L_2$ penalty.
For the $L_2$ penalty, the resulting estimate is a
linear combination of the two within sample  estimates.
Theorem~\ref{thm:df} gives a formula for the degrees of
freedom of that estimate.
Theorem~\ref{thm:pmse} presents the mean squared error
of the estimator and forms the basis for plug-in
estimation of an oracle's value when an
$L_2$ penalty is used.
We also show how to use Stein shrinkage, shrinking the
regression coefficient in the small sample towards the
estimate from the large sample. Such shrinkage makes
it inadmissible to ignore the large sample when
there are $3$ or more coefficients including the intercept.

Section~\ref{sec:mean} considers in detail
the case where the regression simplifies to
estimation of a population mean.
In that setting, we can determine how plug-in,
bootstrap and cross-validation estimates of
tuning parameters behave.  We get an expression
for how much information the large sample can add.
Theorem~\ref{thm:l1closed} gives a soft-thresholding
expression for the estimate produced by $L_1$ penalization
and equation~\eqref{eq:l1loss} can be used to find
the penalty parameter that an $L_1$ oracle would choose
when the data are Gaussian.

Section~\ref{sec:simu} presents some simulated
examples.  We simulate the location problem 
for several $L_2$ penalty methods
varying in how aggressively they use the larger sample. 
The $L_1$ oracle is outperformed by the $L_2$ oracle in this setting.
When the bias is small, 
the data enrichment methods improve upon the small
sample, but when the bias is large then it is
best to use the small sample only.
Things change when we simulate the regression model.
For dimension $d\ge 5$, data enrichment outperforms
the small sample method in our simulations at all bias levels.
We did not see such an inadmissibility outcome
when we simulated cases with $d\le 4$.
In our simulated examples, the data enrichment estimator
performs better than plain Stein shrinkage of the small
sample towards the large sample. 

Section~\ref{sec:prop} presents theoretical support for
our estimator.
Theorem~\ref{thm:plug-in}.
shows that when there are $5$ or more predictors and $10$
or more degrees of freedom for error, then
some of our data enrichment estimators make simply
using the small sample inadmissible.
The reduction in mean squared error is greatest when the bias
is smallest, but no matter how large the bias
is, we gain an improvement. 
The estimator we study employs a data-driven
weighting of the two within-sample least squares estimators.
In simulations, our plug-in estimator performed even better than the
estimator from Theorem~\ref{thm:plug-in}.
Section~\ref{sec:moracle} explains how our estimators are closer
to a matrix oracle than the James-Stein estimators are, and this
may explain why they outperform simple shrinkage in our simulations.



There are many statistical settings where data from one
setting is used to study a different setting. They
range from older methods in survey sampling, to recently
developed methods for bioinformatics.
Section~\ref{sec:lite}
surveys some of those literatures.
Section~\ref{sec:conc} has brief conclusions.
The longer proofs are  in Section~\ref{sec:proofs}
of the Appendix. 

Our contributions include the following:
\begin{compactitem}
\item a new penalization method for combining data sets,
\item an inadmissibility result based on that method, 
\item a comparison of $L_1$ and $L_2$ penalty oracles for the location setting, and
\item evidence that more aggressive shrinkage pays in high dimensions.
\end{compactitem}

\section{Data enriched regression}\label{sec:nota}

Consider linear regression with a response $Y\in \real$
and predictors $X\in\real^d$.  The model for the small
data set is
$$
Y_i = X_i\beta + \err_i,\quad i\in S
$$
for a parameter $\beta\in\real^d$ and 
independent errors $\err_i$ with mean $0$
and variance $\sigma_S^2$.
Now suppose that the data in the big data
set follow
$$
Y_i = X_i(\beta+\gamma) + \err_i,\quad i\in B
$$
where $\gamma\in\real^d$ is a bias parameter
and $\err_i$ are independent with mean $0$
and variance $\sigma^2_B$.
The sample sizes are $n$ in the small sample
and $N$ in the big sample.

There are several kinds of departures of interest.
It could be, for instance, that the overall
level of $Y$ is different in $S$ than in $B$
but that the trends are similar.
That is, perhaps only the intercept component of $\gamma$
is nonzero. More generally, the effects of some but
not all of the components in $X$ may differ
in the two samples.  
One could apply hypothesis testing to each component of $\gamma$
but that is unattractive as the number of
scenarios to test for grows as~$2^d$.

Let $X_S\in\real^{n\times d}$  and $X_B\in\real^{N\times d}$ 
have rows made of vectors $X_i$ for
$i\in S$ and $i\in B$ respectively.
Similarly, let $Y_S\in\real^n$ and $Y_B\in\real^N$ be 
corresponding vectors of response values.
We use $V_S = X_S^\tran X_S$ and $V_B=X_B^\tran X_B$.

\subsection{Partial pooling via shrinkage and weighting}

Our primary approach is to pool the data but put
a shrinkage penalty on $\gamma$.  We estimate
$\beta$ and $\gamma$ by minimizing
\begin{align}\label{eq:penss}
\sum_{i\in S}(Y_i - X_i\beta)^2
+ \sum_{i\in B}(Y_i - X_i(\beta+\gamma))^2
+ \lambda P(\gamma)
\end{align}
where $\lambda\in[0,\infty]$
and $P(\gamma)\ge 0$ is a penalty function.
There are several reasonable choices for the
penalty function, including
$$
\Vert\gamma\Vert_2^2,\quad
\Vert X_S\gamma\Vert_2^2,\quad
\Vert\gamma\Vert_1,\quad\text{and}\quad
\Vert X_S\gamma\Vert_1.
$$
For each of these penalties, setting $\lambda=0$
leads to separate fits $\hat\beta$ and
$\hat\beta+\hat\gamma$ in the two data sets.
Similarly, taking $\lambda=\infty$ constrains
$\hat\gamma=0$ and amounts to pooling the samples.
We will see that varying $\lambda$ shifts the relative
weight applied to the two samples.
In many applications one will want to regularize $\beta$
as well, but in this paper we only penalize $\gamma$.

The criterion~\eqref{eq:penss} does not account
for different variances in the two samples.  If
we knew the variance ratio we might multiply the
sum over $i\in B$ by $\tau = \sigma^2_S/\sigma^2_B$.
A large value of $\tau$ has the consequence of increasing
the weight on the $B$ sample.  Choosing $\tau$ is 
largely confounded with choosing $\lambda$ because
$\lambda$ also adjusts the relative weight of the two samples.
We use $\tau=1$, but in a context with some prior knowledge
about the magnitude of $\sigma^2_S/\sigma^2_B$ another value could be
used. Our inadmissibility result does not depend on knowing
the correct $\tau$.

The $L_1$ penalties have an advantage in interpretation
because they identify which parameters or which 
specific observations might be differentially affected.
The quadratic penalties are more analytically tractable, so we
focus most of this paper on them.

Both quadratic penalties can be expressed
as $\Vert X_T\gamma\Vert_2^2$ for
a matrix $X_T$.  The rows of $X_T$ represent
a hypothetical target population of $N_T$ items for prediction.
The matrix $V_T=X_T^\tran X_T$ is then proportional
to the matrix of mean squares and mean cross-products for predictors
in the target population.

If we want to remove the pooling
effect from one of the coefficients,
such as the intercept term, then the
corresponding column of $X_T$ should contain
all zeros. We can also constrain $\gamma_j=0$
(by dropping its corresponding predictor)
in order to enforce exact pooling on the $j$'th coefficient.

A second, closely related approach is to
fit $\hat\beta_S$ by minimizing
$\sum_{i\in S}(Y_i - X_i\beta)^2$,
fit $\hat\beta_B$ by minimizing
$\sum_{i\in B}(Y_i - X_i\beta)^2$,
and then estimate $\beta$ by
$$
\hat\beta(\omega)=\omega \hat\beta_S + (1-\omega)\hat\beta_B
$$
for some $0\le\omega\le 1$.
In some special cases the estimates indexed
by the weighting parameter $\omega\in[n/(n+N),1]$ 
are a relabeling of the penalty-based estimates
indexed by the parameter $\lambda\in[0,\infty]$. In other
cases, the two families of estimates differ. The weighting
approach allows simpler tuning methods. Although
we find in simulations that the penalization method is
superior, we can prove stronger results about
the weighting approach.

Given two values of $\lambda$
we consider the larger one to be more `aggressive'
in that it makes more use of the big sample
bringing with it the risk of more bias
in return for a variance reduction.
Similarly, aggressive estimators correspond
to small weights $\omega$ on the small target sample.

\subsection{Quadratic penalties and degrees of freedom}

The quadratic penalty takes the form
$ P(\gamma) = \Vert X_T\gamma\Vert^2_2
=\gamma^\tran V_T\gamma$
for a matrix $X_T\in\real^{r\times d}$
and $V_T = X_T^\tran X_T\in\real^{d\times d}$.
The value $r$ is $d$ or $n$ in the examples
above and could take other values in different
contexts.
Our criterion becomes
\begin{align}\label{eq:quadpenalized}
\Vert Y_S - X_S\beta\Vert^2
+ \Vert Y_B - X_B(\beta+\gamma)\Vert^2
+ \lambda \Vert X_T\gamma\Vert^2.
\end{align}
Here and below $\Vert x\Vert$ means the Euclidean norm $\Vert x\Vert_2$.

Given the penalty matrix $X_T$ and a value for $\lambda$,
the penalized sum of squares~\eqref{eq:quadpenalized} is minimized by
$\hat\beta_\lambda$ and $\hat\gamma_\lambda$ satisfying
$$
\cx^\tran\cx
\begin{pmatrix}
\hat\beta_\lambda\\
\hat\gamma_\lambda
\end{pmatrix}
= 
\cx^\tran\cy
$$
where
\begin{align}\label{eq:defXY}
\cx=
\begin{pmatrix}
X_\ssp & 0\\
X_\brp & X_\brp\\
0 & \lambda^{1/2}X_T
\end{pmatrix}\in\real^{(n+N+r)\times 2d},\quad\text{and}\quad
\cy=
\begin{pmatrix}
Y_\ssp \\
Y_\brp \\
0  
\end{pmatrix}.
\end{align}

To avoid uninteresting complications we suppose that the
matrix $\cx^\tran\cx$ is invertible. 
The representation~\eqref{eq:defXY} also underlies
a convenient computational approach to fitting
$\hat\beta_\lambda$ and $\hat\gamma_\lambda$
using $r$ rows of pseudo-data just as one does
in ridge regression.

The estimate $\hat\beta_\lambda$ can be
written in terms of $\hat\beta_S= V_S^{-1}X_S^\tran Y_S$
and $\hat\beta_B=V_B^{-1}X_B^\tran Y_B$ as
the next lemma shows.

\begin{lemma}\label{lem:matrixweights}
Let $X_S$, $X_B$, and $X_T$ in~\eqref{eq:quadpenalized}
all have rank $d$. Then for any $\lambda\ge 0$,
the minimizers $\hat\beta$ and $\hat\gamma$ of~\eqref{eq:quadpenalized}
satisfy
$$\hat\beta =  W_\lambda\hat\beta_S + (I-W_\lambda)\hat\beta_B$$
and $\hat\gamma = (V_B+\lambda V_T)^{-1}V_B(\hat\beta_B-\hat\beta)$ for a matrix
\begin{align}\label{eq:matrixwts}
W_\lambda = (V_S+\lambda V_TV_B^{-1}V_S+\lambda V_T)^{-1}(V_S+\lambda V_TV_B^{-1}V_S).
\end{align}
If $V_T=V_S$, then
$$W_\lambda = (V_B+\lambda V_S+\lambda V_B)^{-1}(V_B+\lambda V_S).$$
\end{lemma}
\begin{proof}
The normal equations of~\eqref{eq:quadpenalized} are
$$
(V_B+V_S)\hat\beta = V_S\hat\beta_S+V_B\hat\beta_B - V_B\hat\gamma
\quad\text{and}\quad
(V_B+\lambda V_T)\hat\gamma =V_B\hat\beta_B - V_B\hat\beta.
$$
Solving the second equation for $\hat\gamma$, plugging the result into
the first and solving for $\hat\beta$, yields the result with
$W_\lambda = (V_S+V_B-V_B(V_B+\lambda V_T)^{-1}V_B)^{-1}V_S$.
This expression for $W_\lambda$ simplifies as given and simplifies
further when $V_T=V_S$.
\end{proof}

The remaining challenge in model fitting
is to choose a value of $\lambda$.
Because we are only interested in making predictions for
the $S$ data, not the $B$ data, 
the ideal value of $\lambda$ is one that
optimizes the prediction error on sample $S$.
One reasonable approach is
to use cross-validation by holding out a portion
of sample $S$ and predicting the held-out values
from a model fit to the held-in ones as well
as the entire $B$ sample.
One may apply either leave-one-out cross-validation
or more general $K$-fold cross-validation. In the latter
case, sample $S$ is split into $K$ nearly equally
sized parts and predictions based on sample $B$
and $K-1$ parts of sample $S$ are used for the $K$'th
held-out fold of sample $S$.

We prefer to use
criteria such as AIC, AICc, or BIC in order to avoid
the cost and complexity of cross-validation. 
To compute AIC and alternatives,
we need to measure the degrees of freedom
used in fitting the model. 
We define the degrees of freedom to be
\begin{align}\label{eq:efdf}
\df(\lambda)=\frac1{\sigma_S^2}\sum_{i\in S}\cov(\hat Y_i,Y_i),
\end{align}
where
$\hat Y_S = X_S\hat\beta_\lambda$.
This is the formula of \cite{ye:1998} and \cite{efro:2004} 
adapted to our setting where the focus is only on predictions
for the $S$ data.
We will see later that the resulting AIC type estimates
based on the degrees of freedom perform similarly
to our focused cross-validation described above.

\begin{theorem}\label{thm:df}
For data enriched regression 
the degrees of freedom given at~\eqref{eq:efdf} satisfies
$\df(\lambda) = \tr(W_\lambda)$ where
$W_\lambda$ is given in Lemma~\ref{lem:matrixweights}.
If $V_T=V_S$, then
\begin{align}\label{eq:df}
\df(\lambda) = 
\sum_{j=1}^d \frac{1+\lambda\nu_j}{1+\lambda+\lambda\nu_j}
\end{align}
where $\nu_1,\dots,\nu_d$ are the eigenvalues of 
\begin{align}\label{eq:eigM}
M\equiv V_S^{1/2}V_B^{-1}V_S^{1/2}
\end{align}
in which $V_S^{1/2}$ is a symmetric matrix square root of $V_S$.
\end{theorem}
\begin{proof}
 Section~\ref{sec:thmdf} in the Appendix.
\end{proof}

With a notion of degrees of freedom customized
to the data enrichment context we can now define
the corresponding criteria such as
\begin{align}
\text{AIC}(\lambda) & = n \log (\hat\sigma^2_S(\lambda)) + n\Bigl(
1+\frac{2\df(\lambda)}n\Bigr)\quad \text{and}\notag\\
\text{AICc}(\lambda) & = 
n \log (\hat\sigma^2_S(\lambda)) + n
\Bigl(1+\frac{\df(\lambda)}n\Bigr)
\Bigm/\Bigl(1-\frac{\df(\lambda)+2}n\Bigr),
\end{align}
where $\hat\sigma^2_{S}(\lambda)
= (n-d)^{-1}\sum_{i \in S}^n(Y_i-X_i\hat\beta(\lambda))^2$.
The AIC is more appropriate than BIC here since
our goal is prediction accuracy, not model selection.
We prefer the AICc criterion of \cite{hurv:tsai:1989}
because it is more conservative as the degrees
of freedom become large compared to the sample size.

Next we illustrate some special cases of
the degrees of freedom formula in Theorem~\ref{thm:df}.
First, suppose that $\lambda =0$, so that there is no penalization
on $\gamma$.  Then 
$\df(0) = \tr( I ) = d$ as is appropriate for regression on
sample $S$ only.

We can easily see that the degrees of freedom are
monotone decreasing in $\lambda$. As $\lambda\to\infty$
the degrees of freedom drop to
$\df(\infty) = \sum_{j=1}^d\nu_j/(1+\nu_j)$.
This can be much smaller than $d$. For instance
if $V_S=n\Sigma$ and $V_B=N\Sigma$ for some positive
definite $\Sigma\in\real^{d\times d}$,
then all $\nu_j = n/N$ and so
$\df(\infty) = d/(1+N/n)\le dn/N$.


Monotonicity of the degrees of freedom makes it
easy to search for the value $\lambda$ which
delivers a desired degrees of freedom.  We have found
it useful to investigate $\lambda$ over a numerical
grid corresponding to degrees of freedom decreasing
from $d$ by an amount $\Delta$ (such as $0.25$)
to the smallest such value above $\df(\infty)$.
It is easy to adjoin $\lambda=\infty$ (sample pooling)
to this list as well.

\subsection{Predictive mean square errors}

Here we develop an oracle's choice for $\lambda$ and
a corresponding plug-in estimate.
We work in the  case where $V_S=V_T$ and we assume
that $V_S$  has full rank.
Given $\lambda$,
the predictive mean square error is $\e(\Vert X_{S}(\hat{\beta}-\beta)\Vert^{2})$. 

We will use the matrices $V_S^{1/2}$ and $M$ from
Theorem~\ref{thm:df}
and  the eigendecomposition $M= UDU^\tran$
where the $j$'th column of $U$ is $u_j$
and $D= \diag(\nu_j)$.

\begin{theorem}\label{thm:pmse}
The predictive mean square error  of the
data enrichment estimator is
\begin{align}
\e\bigl(\Vert X_{S}(\hat{\beta}-\beta)\Vert^2\bigr)
 & =\sigma_{S}^{2}\sum_{j=1}^{d}\frac{(1+\lambda\nu_{j})^{2}}{(1+\lambda+\lambda\nu_{j})^{2}}+
\sum_{j=1}^{d}\frac{\lambda^2\kappa_j^2}{(1+\lambda+\lambda\nu_{j})^{2}}
\label{eq:predictive-mse}\end{align}
where 
$\kappa_j^2 = {u_{j}^\tran V_{S}^{1/2}\Theta V_{S}^{1/2}u_{j}}$
for
$\Theta=\gamma\gamma^{T}+\sigma_{B}^{2}V_{B}^{-1}$.
\end{theorem}
\begin{proof}
 Section~\ref{sec:pmse}.
\end{proof}

The first term in~\eqref{eq:predictive-mse} is a variance
term.  It equals $d\sigma^2_S$ when $\lambda =0$
but for $\lambda>0$ it is reduced due to the use
of the big sample.
The second term represents the error, both bias squared
and variance, introduced by the big sample.

\subsection{A plug-in method}

A natural choice of $\lambda$ is to minimize the predictive mean
square error, which must be estimated. We propose a plug-in method
that replaces the unknown parameters $\sigma_{S}^{2}$ and $\kappa_j^2$
from Theorem~\ref{thm:pmse} by sample estimates.
For estimates ${\hat\sigma_{S}^{2}}$ and $\hat\kappa_j^2$
we choose
\begin{align}
\hat{\lambda} &
=\arg\min_{\lambda\geq0}
\,\sum_{j=1}^{d}\,
\frac{\hat\sigma_{S}^{2}
(1+\lambda\nu_j)^2 + \lambda^2\hat\kappa_j^2}
{(1+\lambda+\lambda\nu_{j})^{2}}.
\label{eq:plug-in-lambda}
\end{align}

From the sample data we take 
$\hat\sigma^2_S = \Vert Y_S-X_S\hat\beta_S\Vert^2/(n-d)$.
A straightforward plug-in estimate of 
the matrix $\Theta$ in Theorem~\ref{thm:pmse} is 
\begin{align*}
\widehat{\Theta}_{\plug} & =\hat{\gamma}\hat{\gamma}^\tran+{\hat\sigma_{B}^{2}}V_{B}^{-1},
\end{align*}
where $\hat\gamma = \hat\beta_B - \hat\beta_S$.
Now we take $\hat\kappa^2_j = u_j^\tran V_S^{1/2}\wh\Theta V_S^{1/2}u_j$
recalling that $u_j$ and $\nu_j$ derive from
the eigendecomposition of $M=V_S^{1/2}V_B^{-1}V_S^{1/2}$.
The resulting optimization yields an estimate $\hat\lambda_{\plug}$.

The estimate $\wh\Theta_\plug$ is biased upwards because
$\e( \hat\gamma\hat\gamma^\tran)
=\gamma\gamma^\tran + \sigma^2_BV_B^{-1}+\sigma^2_SV_S^{-1}$.
We have used a bias-adjusted plug-in estimate
\begin{align}\label{eq:bapi}
\widehat\Theta_\bapi 
= \hat\sigma_B^2 V_B^{-1}+ (\hat\gamma\hat\gamma^\tran-\hat\sigma_B^2
V_B^{-1} - \hat\sigma_S^2V_S^{-1})_+
\end{align}
where the positive part operation on a symmetric matrix preserves
its eigenvectors but replaces any negative eigenvalues by $0$.
Similar results can be obtained with
\begin{align}\label{eq:bapi2}
\wt{\Theta}_{\bapi}
=\bigl(\hat{\gamma}\hat{\gamma}^\tran-{\hat\sigma_{S}^{2}}V_{S}^{-1}\bigr)_{+}.
\end{align}
This estimator is somewhat simpler but~\eqref{eq:bapi}
has the advantage of being at least as large as
$\hat\sigma^2_BV_B^{-1}$ while~\eqref{eq:bapi2} can degenerate to $0$.

\subsection{James-Stein shrinkage estimators}\label{sec:stein}

For background on these estimators, see \cite{efron1973stein}.
We shrink $\hat\theta_S = V_S^{1/2}\hat\beta_S
\sim\dnorm(V_S^{1/2}\beta,\sigma^2_SI_n)$ towards
a target vector, to get better estimators of 
$\theta_S = V_S^{1/2}\beta$.
To make use of the big data set we shrink $\hat\theta_S$ towards
$$
\hat\theta_B = V_S^{1/2}\hat\beta_B \sim\dnorm( V_S^{1/2}(\beta+\gamma),
V_S^{1/2}V_B^{-1}V_S^{1/2}\sigma^2_B). 
$$

We consider two shrinkers
\begin{align}
\hat\theta_{\js,B} &= \hat\theta_B
+\Bigl( 1- \frac{d-2}{\Vert\hat\theta_S-\hat\theta_B\Vert^2/\sigma^2_S}\Bigr)(\hat\theta_S-\hat\theta_B),\quad\text{and}\notag\\
\hat\theta_{\js,B+} &= \hat\theta_B
+\Bigl( 1- \frac{d-2}{\Vert\hat\theta_S-\hat\theta_B\Vert^2/\sigma^2_S}\Bigr)_+(\hat\theta_S-\hat\theta_B).\label{eq:jsbp}
\end{align}
Each of these makes $\hat\theta_S$ inadmissible in squared
error loss as an estimate of $\theta_S$, when $d\ge3$.
The squared error loss on the $\theta$ scale is
\begin{align}\label{eq:theloss}
(\hat\theta_S-\theta_S)^\tran (\hat\theta_S-\theta_S)
=(\hat\beta_S-\beta_S)^\tran V_S(\hat\beta_S-\beta_S).
\end{align}

When $d\ge3$ and our quadratic loss is based on $V_S$,
we can make $\hat\beta_S$ inadmissible by shrinkage,
so long as $d\ge3$. \cite{copa:1983} 
found that ordinary least squares regression is
inadmissible when $d\ge4$. 
\cite{stei:1960} also obtained an inadmissibility result
for regresssion, but under stronger conditions than Copas
needs. \cite{copa:1983} applies no shrinkage to the intercept but shrinks
the rest of the coefficient vector towards zero.
In this problem it is reasonable to shrink the entire coefficient vector
as the big data set supplies a nonzero default intercept.

\section{The location model}\label{sec:mean}

The simplest instance of our problem is the
location model where $X_S$ is a column of $n$
ones and $X_B$ is a column of $N$ ones.
Then the vector $\beta$
is simply a scalar intercept that we call $\mu$
and the vector $\gamma$ is a scalar mean
difference that we call $\delta$.
The response values in the
small data set are $Y_i=\mu+\err_i$
while those in the big data set are $Y_i =(\mu+\delta)+\err_i$.
In the location family we lose no generality taking the
quadratic penalty to be $\lambda\delta^2$.

The quadratic criterion is
$
\sum_{i\in S} (Y_i-\mu)^2 + \sum_{i\in B}(Y_i-\mu-\delta)^2 + \lambda\delta^2.
$
Taking $V_S=n$, $V_B=N$ and $V_T=1$ in Lemma~\ref{lem:matrixweights}
yields
$$
\hat\mu=
\omega \bar Y_S + (1-\omega)\bar Y_B
\quad\text{with}\quad
\omega = \frac{nN+n\lambda}{nN+n\lambda+N\lambda}
= \frac{1+\lambda/N}{1+\lambda/N+\lambda/n}.$$
Choosing  a value for $\omega$ corresponds
to choosing 
$$\lambda = \frac{nN(1-\omega)}{N\omega-n(1-\omega)}.$$
The degrees of freedom in this case reduce to
$\df(\lambda) =\omega$,
which ranges from $\df(0)=1$ down to
$\df(\infty) = n/(n+N)$.

\subsection{Oracle estimator of $\omega$}

The mean square error of $\hat\mu(\omega)$ is
$$
\mse(\omega)
=\omega^2\frac{\sigma^2_S}n
+(1-\omega)^2\Bigl( \frac{\sigma^2_B}N
+\delta^2\Bigr).
$$
The mean square optimal value of $\omega$ (available to an oracle) is
$$
\omega_\orcl = 
\frac{\delta^2 + \sigma_B^2/N}{\delta^2 + \sigma_B^2/N+\sigma^2_S/n}.
$$
Pooling the data corresponds
to $\omega_\pool = n/(N+n)$
and makes $\hat\mu$ equal the pooled
mean $\bar Y_P\equiv (n\bar Y_S+N\bar Y_B)/(n+N)$.
Ignoring the large data set 
corresponds to $\omega_S=1$.
Here $\omega_\pool \le \omega_\orcl \le \omega_S$.

The mean squared error reduction for the oracle is
\begin{align}\label{eq:orclmse}
\frac{\mse(\omega_\orcl)}{\mse(\omega_S)}
= \omega_\orcl,
\end{align}
after some algebra.
If $\delta\ne0$, then as $\min(n,N)\to\infty$
we find $\omega_{\orcl}\to 1$ and the optimal $\omega$
corresponds to simply using the small sample and ignoring
the large one.
If instead $\delta\ne0$ and $N\to\infty$ for finite $n$, then
the effective sample size for data enrichment may be
defined using~\eqref{eq:orclmse} as
\begin{align}\label{eq:degain}
\wt n =\frac{n}{\omega_\orcl} = n\frac{\delta^2+\sigma^2_B/N+\sigma^2_S/n}{\delta^2+\sigma^2_B/N}
\to n + \frac{\sigma^2_S}{\delta^2}.
\end{align}
The mean squared error from data enrichment with $n$ observations
in the small sample, using the oracle's choice of $\lambda$, 
matches that of $\wt n$ IID observations from the small sample.
We effectively gain up to  $\sigma^2_S/\delta^2$ observations
worth of information. This is an upper bound on the gain because
we will have to estimate~$\lambda$.

Equation~\eqref{eq:degain} shows that
the benefit from data enrichment is a small sample
phenomenon.  The effect is additive not multiplicative
on the small sample size $n$.  As a result, more valuable
gains are expected in small samples.
In some of the motivating examples we have found
the most meaningful improvements from data enrichment on
disaggregated data sets, such as specific groups of
consumers. 

\subsection{Plug-in and other estimators of $\omega$}

A natural approach to choosing $\omega$ is
to plug in sample estimates
$$
\hat\delta_0 = \bar Y_B-\bar Y_S,\quad
\hat\sigma^2_S = \frac1n\sum_{i\in S}(Y_i-\bar Y_S)^2,\quad
\text{and}\quad
\hat\sigma^2_B = \frac1N\sum_{i\in B}(Y_i-\bar Y_B)^2.
$$
We then use
$\omega_\plug=
({\hat\delta_0^2 + \hat\sigma_B^2/N})/
({\hat\delta_0^2 + \hat\sigma_B^2/N+\hat\sigma^2_S/n})
$
or equivalently
$\lambda_\plug=
\hat\sigma^2_S/(\hat\delta_0^2 + (\hat\sigma^2_B-\hat\sigma^2_S)/N)$.
Our bias-adjusted plug-in method reduces to
\begin{align*}
\omega_{\bapi}
&
= \frac{\hat\theta_\bapi}{\hat\theta_\bapi +
  \hat\sigma^2_S/n},\quad\text{where}\quad
\hat\theta_\bapi =  \frac{\hat\sigma^2_B}N + 
\Bigl(\hat\delta_0^2 - \frac{\hat\sigma_S^2}n -  \frac{\hat\sigma_B^2}N\Bigr)_+.
\end{align*} 
The simpler alternative
$\wt\omega_{\bapi}
=((\hat\delta_0^2 - \hat\sigma_S^2/n)/\hat\delta_0^2)_+$
gave virtually identical values in our numerical results reported below.

If we bootstrap the $S$ and $B$ samples
independently $M$ times and choose $\omega$
to minimize
$$
\frac1M\sum_{m=1}^M
\bigl(
\bar Y_S -
\omega \bar Y_S^{m*}-(1-\omega)\bar Y_B^{m*}
\bigr)^2,
$$
then the minimizing value tends to $\omega_\plug$
as $M\to\infty$. Thus bootstrap methods give an
approach analogous to plug-in methods, when no
simple plug-in formula exists.  


We can also determine the effects of cross-validation
in the location setting, and arrive at an estimate
of $\omega$ that we can use without actually cross-validating.
Consider splitting the small sample
into $K$ parts that are held out one by one
in turn. The $K-1$ retained parts are used
to estimate $\mu$ and then the squared
error is judged on the held-out part.
That is
$$
\omega_\cv 
=
\arg\min_\omega
\frac1K\sum_{k=1}^K
\bigl(\bar Y_{S,k}-\omega\bar Y_{S,-k}-(1-\omega)\bar Y_B\bigr)^2,
$$
where $\bar Y_{S,k}$ is the average of $Y_i$ over the $k$'th
part of $S$ and $\bar Y_{S,-k}$ is the average of $Y_i$ over
all $K-1$ parts excluding the $k$'th. 

If $n$ is a multiple of $K$ and we average over
all of the $K$-fold sample splits we might use,
then an analysis in Section~\ref{sec:cvderive}
shows that $K$-fold cross-validation chooses a
weighting centered around
\begin{align}\label{eq:cvk}
\omega_{\cv,K} & = 
\frac
{ \hat\delta_0^2 - \hat\sigma^2_{S}/(n-1)}
{ \hat\delta_0^2 + \hat\sigma^2_{S}/[(n-1)(K-1)]}.
\end{align}
Cross-validation allows $\omega<0$. This can
arise when the bias is small and then sampling
alone makes the held-out part of the small sample
appear negatively correlated with the held-in part.
The effect can appear with any $K$. 
We replace any $\omega_{\cv,K} <n/(n+N)$ by $n/(n+N)$.

Leave-one-out cross-validation has $K=n$ (and $r=1$) so it
chooses a weight centered around
$\omega_{\cv,n} = 
[{\hat\delta_0^2 - \hat\sigma^2_S/(n-1)}]/
[\hat\delta_0^2+\hat\sigma^2_S/(n-1)^2]$.
Smaller $K$, such as choosing $K=10$ versus $n$,
tend to make $\omega_{\cv,K}$ smaller resulting
in less weight on $\bar Y_S$. 
In other words, $10$-fold cross-validation makes more aggressive
use of the large sample than does leave-one-out.

\begin{remark}
The cross-validation estimates do not make use of $\hat\sigma^2_B$
because the large sample is held fixed. They are in this sense
conditional on the large sample.  Our oracle takes account of the
randomness in set $B$, so it is not conditional.  One can define
a conditional oracle without difficulty, but we omit the details.
Neither  the bootstrap nor the plug-in methods are conditional, 
as they approximate our oracle.
Taking $\omega_\bapi$ as a representor of unconditional methods
and $\omega_{\cv,n}$ as a representor of conditional ones,
we see that the latter has a larger denominator while they
both have the same numerator, at least when $\hat\delta^2_0>\hat\sigma^2_S/n$.
This suggests that conditional methods are more aggressive and
we will see this in the simulation results.
\end{remark}

\subsection{$L_1$ penalty}

For the location model,
it is convenient to write the $L_1$
penalized criterion as
\begin{align}\label{eq:l1loccrit}
\sum_{i\in S}(Y_i-\mu)^2+\sum_{i\in B}(Y_i-\mu-\delta)^2 + 2\lambda|\delta|.
\end{align}
The minimizers $\hat\mu$ and $\hat\delta$ satisfy
\begin{equation}\label{eq:l1locest}
\begin{split}
\hat\mu &= \frac{n\bar Y_S + N(\bar Y_B-\hat\delta)}{n+N},\quad\text{and}\\
\hat\delta & = \Theta( \bar Y_B-\hat\mu;\lambda/N)
\end{split}
\end{equation} 
for the 
soft thresholding function
$\Theta( z;\tau) = \sign(z)(|z|-\tau)_+$.

The estimate $\hat\mu$ ranges from $\bar Y_S$
at $\lambda=0$ to the pooled mean $\bar Y_P$ at $\lambda=\infty$.
In fact $\hat\mu$ reaches $\bar Y_P$ at 
a finite value $\lambda = \lambda_*\equiv nN|\bar Y_B-\bar Y_S|/(N+n)$
and both $\hat\mu$ and $\hat\delta$ are linear in $\lambda$ 
on the interval $[0,\lambda_*]$:

\begin{theorem}\label{thm:l1closed}
If $0\le\lambda\le nN|\bar Y_B-\bar Y_S|/(n+N)$
then the minimizers of~\eqref{eq:l1loccrit} are
\begin{equation}\label{eq:l1soln}
\begin{split}
\hat\mu & = \bar Y_S +
\frac{\lambda }{n}\, \sign(\bar Y_B-\bar Y_S) ,\quad\text{and}\\
\hat\delta & = \bar Y_B-\bar Y_S - \lambda\frac{N+n}{Nn}\,\sign(\bar Y_B-\bar Y_S).
\end{split}
\end{equation}
If $\lambda >nN|\bar Y_B-\bar Y_S|/(n+N)$ then they are
$\hat\delta = 0$ and $\hat\mu = \bar Y_P$.
\end{theorem}
\begin{proof}
 Section~\ref{sec:proofl1closed} in the Appendix.
\end{proof}

With an $L_1$ penalty on $\delta$ we find from
Theorem~\ref{thm:l1closed} that
$$
\hat\mu = \bar Y_S + \min(\lambda,\lambda_*) \sign(\bar Y_B-\bar Y_S)/n.
$$
That is, the estimator moves $\bar Y_S$ towards $\bar Y_B$
by an amount $\lambda/n$ except that it will not
move past the pooled average $\bar Y_P$.
The optimal choice of $\lambda$ is not available in
closed form. 

\subsection{An $L_1$ oracle}

The $L_2$ oracle depends only on moments of the data.
The $L_1$ case proves to be more complicated, depending
also on quantiles of the error distribution.
To investigate $L_1$ penalization, 
we suppose that the errors are Gaussian.
Then we can compute $\e((\hat\mu(\lambda)-\mu)^2)$ 
for the $L_1$ penalization by a lengthy expression
broken into several steps below. That expression is not simple to interpret.
But we can use it to numerically find the best value
of $\lambda$ for an oracle using the $L_1$ penalty.  
That then allows us to compare the $L_1$ and $L_2$ oracles
in Section~\ref{sec:locationnumer}.

Let $D=\bar Y_B-\bar Y_S$
and then define
\begin{align*}
F & = N/(n+N) & c & = \lambda/(nF)\\
\tau & = (\sigma^2_S/n+\sigma^2_B/N)^{1/2}
&\alpha & = (\sigma^2_B/N)/(\sigma^2_S/n)\\
\eta_+ & = (\delta-c)/\tau & \eta_- & = (-\delta-c)/\tau\\
\varphi_\pm & = \varphi(\eta_\pm) & \Phi_\pm & = \Phi(\eta_\pm)\\
c_0 & = \alpha/(\alpha+1)+F-1,\quad\text{and} & c_1 &= 1/(\alpha+1).
\end{align*}
Lemma~\ref{lem:dparts} in the Appendix gives these identities:
\begin{align*}
\e( 1_{|D|\ge c}) & = \Phi_+ + \Phi_-\\
\e( D1_{|D|\ge c} )
& = 
\delta(\Phi_++\Phi_-) + \tau(\varphi_+-\varphi_-) 
\\
\e( D^21_{|D|\ge c} )
& = 
(\delta^2+\tau^2)(\Phi_++\Phi_-)
+\tau c(\varphi_++\varphi_-)
+\tau \delta(\varphi_+-\varphi_-)
\\
\e( \sign(D)1_{|D|\ge c} ) & = \Phi_+-\Phi_-,\quad\text{and}
\\
\e( D\,\sign(D)1_{|D|\ge c} ) & = 
\delta(\Phi_+ -\Phi_-)
+\tau( \varphi_+ +\varphi_-).
\end{align*}
Section~\ref{sec:proofl1orcl} of the Appendix shows that
\begin{align}\label{eq:l1loss}
\begin{split}
\e((\hat\mu-\mu)^2)
& =
F^2\delta^2
+c_0^2\tau^2 + c_1^2(\alpha^2\sigma^2_S/n + \sigma^2_B/N)\\
& \quad +(\lambda/n)^2(\Phi_++\Phi_-) - 2c_1\delta F\e\bigl(D\onedc)\\
&\quad + F(F-2c_0)\e\bigl(D^2\onedc\bigr)\\
&\quad + 2c_1\delta(\lambda/n)\e\bigl(\sign(D)\onedc\bigr)\\
&\quad + 2(\lambda/n)(c_0-F)\e\bigl(D\,\sign(D)\onedc\bigr).
\end{split}
\end{align}
Substituting the quantities above into~\eqref{eq:l1loss} yields
a computable expression for the loss in the $L_1$ penalized case.

\section{Numerical examples}\label{sec:simu}

We have simulated some special cases
of the data enrichment problem.
First we simulate the pure location problem
which has $d=1$.
Then we consider  the regression problem with varying $d$.

\subsection{Location}\label{sec:locationnumer}

We simulated Gaussian data for the location
problem. The large sample had $N=1000$
observations and the small sample had $n=100$
observations:
$X_i\sim\dnorm(\mu,\sigma^2_S)$ for $i\in S$
and
$X_i\sim\dnorm(\mu+\delta,\sigma^2_B)$ for $i\in B$.
Our data had $\mu=0$ and $\sigma^2_S=\sigma^2_B=1$.
We define the relative bias as
$$
\delta_*=\frac{|\delta|}{\sigma_S/\sqrt{n}} = \sqrt{n}|\delta|.
$$
We investigated a range of relative bias values.
It is only a small simplification to take $\sigma^2_S=\sigma^2_B$.
Doubling $\sigma^2_B$ has a very similar effect to halving $N$.
Equal variances might have given a slight relative advantage to a
hypothesis testing method as described below.

The accuracy of our estimates is judged by
the relative mean squared error
$\e( (\hat\mu-\mu)^2)/(\sigma^2_S/n)$.
Simply taking $\hat\mu =\bar Y_S$ attains
a relative mean squared error of $1$.

Figure~\ref{fig:locationresults} plots relative mean squared error
versus relative bias for a collection of estimators, with the
results averaged over $10{,}000$ simulated data sets.
We used the small sample only method as a control
variate.

\begin{figure}[t]\centering
\includegraphics[width=.9\hsize]
{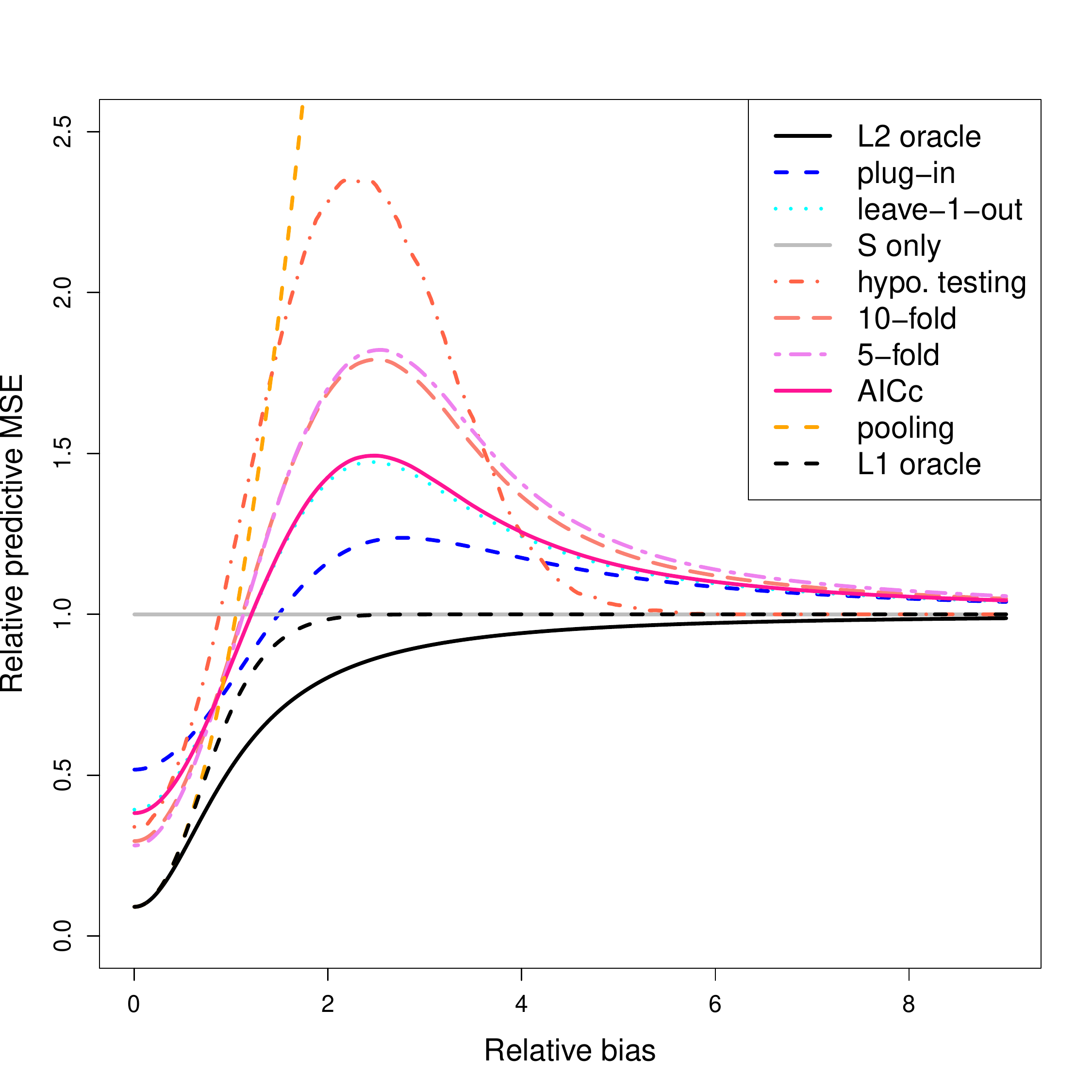}
\caption{\label{fig:locationresults}
Numerical results for the location problem. The horizontal
line at $1$ represents using the small sample only
and ignoring the large one.  The lowest line shown is
for an oracle choosing $\lambda$ in the $L_2$ penalization.
The dashed black curve shows an oracle using the $L_1$ penalization.
The other curves are as described in the text.
}
\end{figure}

The solid curve in Figure~\ref{fig:locationresults}
shows the oracle's value.  It lies strictly below
the horizontal $S$-only line.  None of the competing
curves lie strictly below that line. None can because
$\bar Y_S$ is an admissible estimator for $d=1$
\citep{stei:1956}.
The second lowest curve in Figure~\ref{fig:locationresults}
is for the oracle using the $L_1$ version of the penalty.
The $L_1$ penalized oracle is not as effective as the
$L_2$ oracle and it is also more difficult to approximate.
The highest observed predictive MSEs come from a method
of simply pooling the two samples.  That method is
very successful when the relative bias is near zero but
has an MSE that becomes unbounded as the relative bias
increases.

Now we discuss methods that use the data
to decide whether to use the small sample
only, pool the samples or choose an amount
of shrinkage. 
We may list them in order of their worst
case performance. From top (worst) to
bottom (best) in Figure~\ref{fig:locationresults}
they are: hypothesis testing, $5$-fold cross-validation,
$10$-fold cross-validation, AICc, 
leave-one-out cross-validation, 
and then the simple plug-in method which is minimax among
this set of choices.  AICc and leave-one-out are very close.
Our cross-validation estimators used $\omega =
\max(\omega_{\cv,K},n/(n+N))$
where $\omega_{\cv,K}$ is given by~\eqref{eq:cvk}.

The hypothesis testing method is based on a two-sample $t$-test
of whether $\delta=0$. If the test is rejected at $\alpha = 0.05$,
then only the small sample data is used. If the test is not
rejected, then the two samples are pooled. That test was based
on $\sigma^2_B=\sigma^2_S$ which may give hypothesis
testing a slight advantage in this setting (but it still performed
poorly).

The AICc method performs virtually identically to leave-one-out
cross-validation over the whole range of relative biases.

None of these methods makes any other one
inadmissible: each pair of curves crosses.
The methods that do best at large relative biases
tend to do worst at relative bias near $0$ and
vice versa.  The exception is hypothesis testing.
Compared to the others 
it does not benefit fully from low relative bias
but it recovers the quickest as the bias increases.
Of these methods hypothesis testing
is best at the highest relative bias,
$K$-fold cross-validation with small $K$
is best at the lowest relative bias,
and  the plug-in method is best in between.

Aggressive methods will do better at low bias
but worse at high bias.  What we see in this simulation
is that $K$-fold cross-validation is the most aggressive
followed by leave-one-out and AICc and that plug-in
is least aggressive.  These findings confirm what we
saw in the formulas from Section~\ref{sec:mean}. 
Hypothesis testing does not quite
fit into this spectrum: its worst case performance
is much worse than the most aggressive methods
yet it fails to fully benefit from pooling when the
bias is smallest.  Unlike aggressive methods it does
very well at high bias.

\subsection{Regression}

We simulated our data enrichment method for
the following scenario. 
The small sample had $n=1000$ observations
and the large sample had  $N=10{,}000$. 
The true $\beta$ was taken to be $0$.
This is no loss of generality because we are not
shrinking $\beta$ towards $0$.
The value of $\gamma$ was taken uniformly
on the unit sphere in $d$ dimensions and then
multiplied by a scale factor that we varied.

We considered $d=2,4,5$ and $10$. All of our
examples included an intercept column of $1$s
in both $X_S$ and $X_B$. The other $d-1$ predictors
were sampled from a Gaussian distribution
with covariance $C_S$  or $C_B$, respectively.

In one simulation we took $C_S$ and $C_B$
to be independent Wishart$(I,d-1,d-1)$ random matrices.
In the other simulation, they were sampled as a spiked
covariance model \citep{john:2001}.
There
$C_S = I_{d-1} + \rho uu^\tran$ and
$C_B = I_{d-1} +\rho vv^\tran$
where $u$ and $v$ are independently
and uniformly sampled from the unit sphere in $\real^{d-1}$
and $\rho\ge0$ is a parameter that measures
the lack of proportionality between covariances. 
We chose $\rho=d$ so that
the sample specific portion of the variance
has comparable magnitude to the common part.

The variance in the small sample was $\sigma^2_S=1$.
To model the lower quality of the  large sample we
used $\sigma^2_B=2$.

We scaled the results so that regression using
sample $S$ only yields a mean squared error
of $1$.
We computed the risk of an $L_2$ oracle,
as well as sampling errors when $\lambda$
is estimated by the plug-in formula, by our
bias-adjusted plug-in formula and via AICc.
In addition we considered the simple weighted
combination $\omega\hat\beta_S+(1-\omega)\hat\beta_B$
with $\omega$ chosen by the plug-in formula.
To optimize~\eqref{eq:plug-in-lambda} over $\lambda$
we used the optimize function in R which is based
on golden section search \citep{bren:1973}.

\begin{figure}[t]
\centering 
\includegraphics[width=0.99\hsize]
{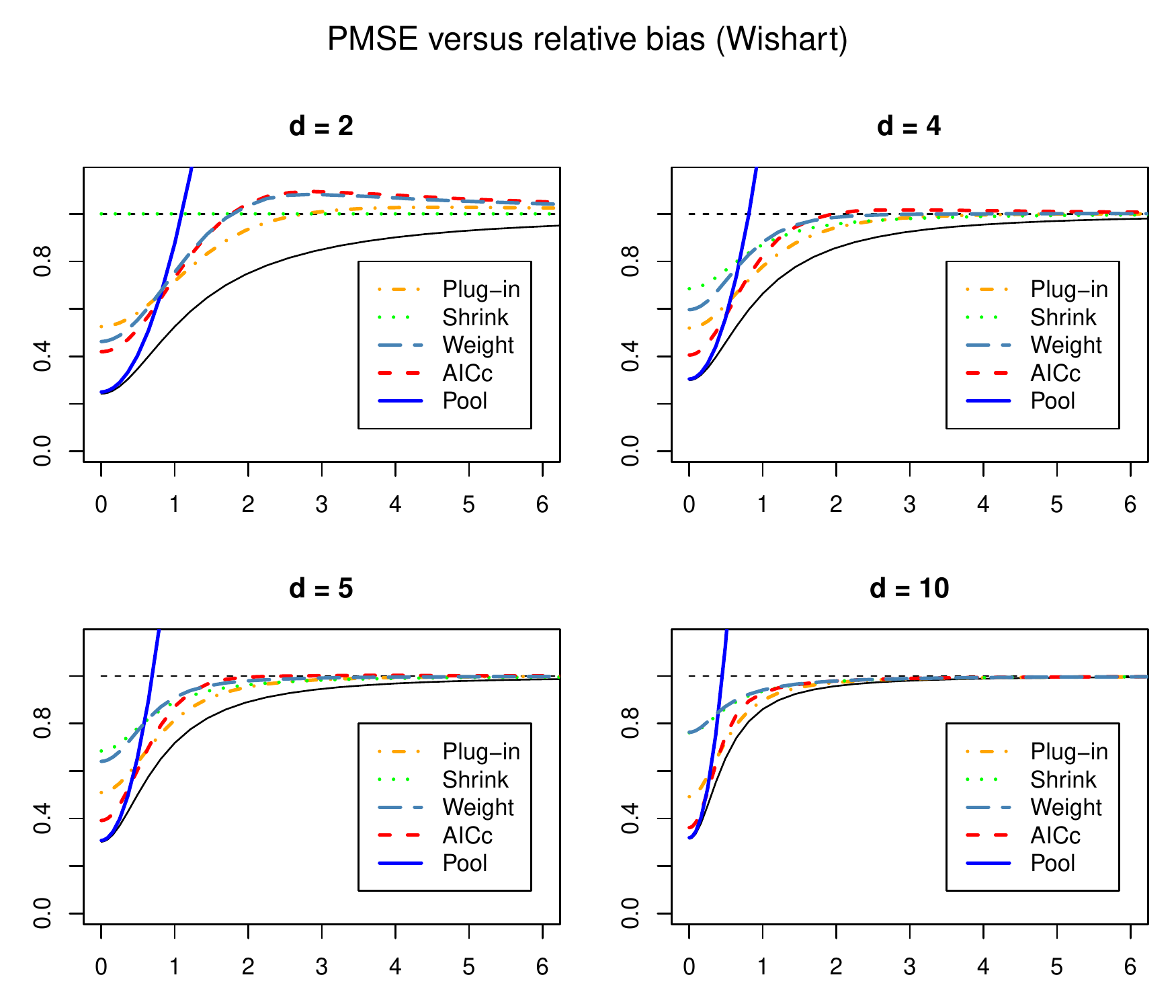}
\caption{\label{fig:regresults0}
Predicted MSE versus
relative bias for the Wishart covariances
described in the text. 
On this scale the small
sample only has MSE one (horizontal dashed line).
Five methods are shown. The lowest curve is for
the oracle. 
}
\end{figure}

\begin{figure}[t]
\centering 
\includegraphics[width=0.99\hsize]
{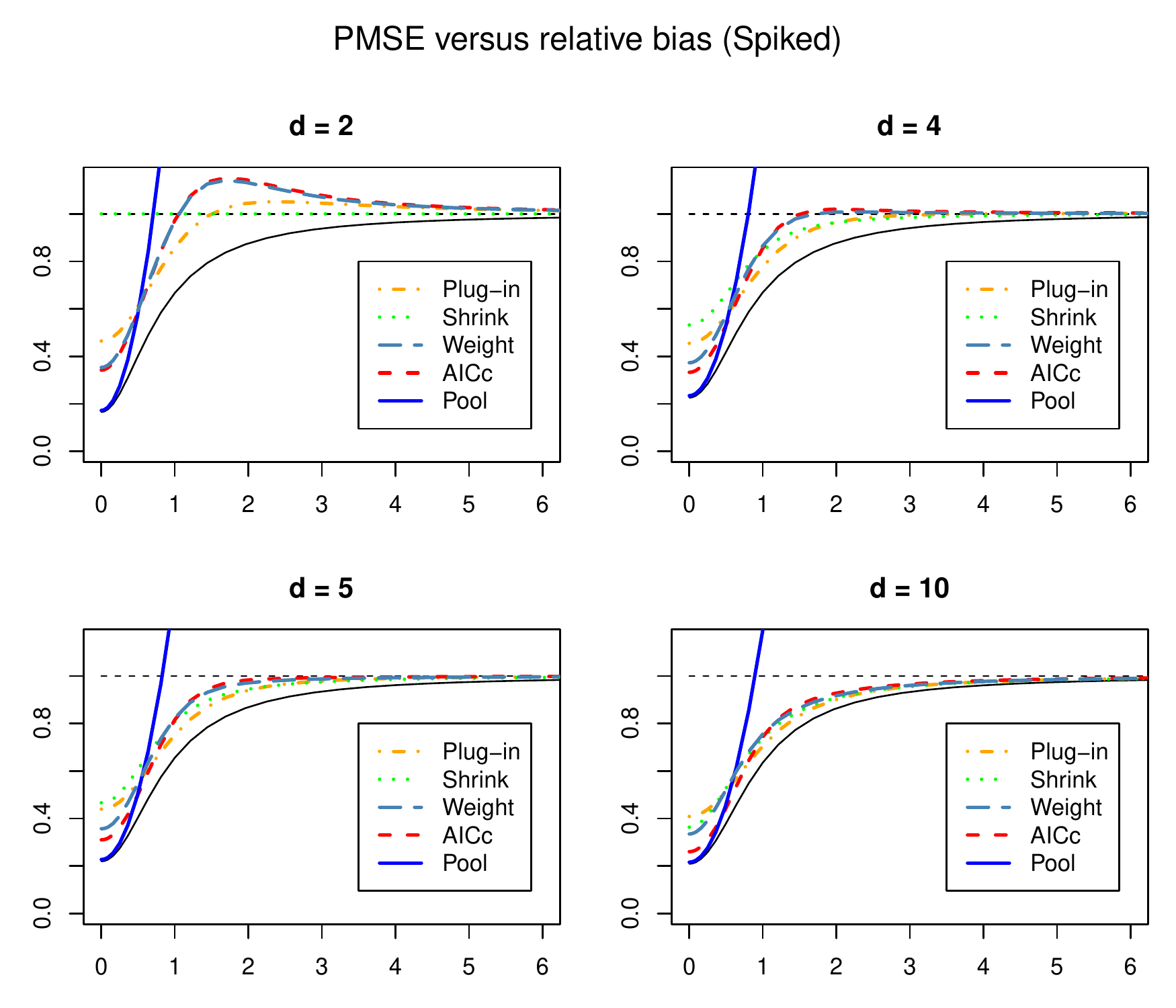}
\caption{\label{fig:regresults1}
Predicted MSE versus
relative bias for the spiked covariances
described in the text. 
On this scale the small
sample only has MSE one (horizontal dashed line).
Five methods are shown. The lowest curve is for
the oracle. 
}
\end{figure}

We also included a shrinkage estimator.
Because our simulated runs all had $\beta_S=0$
it is not reasonable to include shrinkage
of $\hat\beta_S$ towards zero in the comparison;
we cannot in practice shrink towards the truth.
Instead, we used the positive part Stein shrinkage
estimate~\eqref{eq:jsbp}
shrinking $\hat\beta_S$ towards $\hat\beta_B$ but not past it.
That shrinkage requires an estimate $\hat\sigma^2_S$
of $\sigma^2_S$.
We used the true value, $\sigma^2_S=1$, giving the
shrinkage estimator a slight advantage.

We did not include hypothesis testing in this example,
because there are $2^d$ possible ways to decide which
parameters to pool and which to estimate separately.

We simulated each of the two covariance models
with each of the four dimensions $10{,}000$ times.
For each method we averaged the squared prediction
errors $(\hat\beta-\beta)^\tran V_S(\hat\beta-\beta)$
and then divided those mean squared errors by
the one for using the small sample only.
Figures \ref{fig:regresults0} and~\ref{fig:regresults1}
show the results. 

At small bias levels pooling the samples is almost
as good as the oracle.  But the loss for pooling samples
grows without bound when the bias increases.
For $d=2$, shrinkage amounts to using the small sample
only but for $d>2$ it performs universally better
than the small sample.

When comparing methods we see that the curves usually
cross. Methods that are best at low bias tend not to
be best at high bias. Note however that there is a lot
to gain at low bias while the methods differ only a little
at high bias.  As a result the more aggressive methods
making greater use of the large data are more likely to 
yield a big improvement.

The weighting estimator generally performs
better than the shrinkage estimator in that it
offers a meaningful improvement at low bias
costing a minor relative loss at high bias. We analyze
that estimator in Section~\ref{sec:prop}.
The plug-in estimator also is generally better
than shrinkage. The AICc estimator is generally
better than both of those. We do not show the
bias adjusted plug-in estimators.  Their performance
is almost identical to AICc (always within about $0.015$).
Of those, the one using $\wh\Theta_{\mathrm{bapi}}$
was consistently at least as good as $\wt\Theta_{\mathrm{bapi}}$
and sometimes a little better.
 
The greatest gains are at or near zero bias.
Table~\ref{tab:atzero} shows the quadratic losses 
at $\delta=0$ normalized by the loss attained by the
oracle. Pooling is almost as good as the oracle in this
case but we rule it out because of its extreme bad performance
when the bias is large. Some of our new estimators yield
very much reduced squared error compared to the shrinkage
estimator. For example three of the new methods' squared errors are just less
than half that of the shrinkage estimator for $d=10$
and the Wishart covariances.

\newcommand{\therowspace}{1pt}
\begin{table}
\caption{\label{tab:atzero}
This table shows the quadratic loss~\eqref{eq:theloss} 
normalized by that of the oracle when $\delta=0$ (the no bias condition).
The methods are described in the text. 
}\centering
\begin{tabular}{l|cccc|cccc}
\toprule 
 & \multicolumn{4}{c}{Wishart} &
\multicolumn{4}{c}{Spiked}\\
Method & 2 & 4 & 5 & 10 & 2 & 4 & 5 & 10\\
\midrule
[\therowspace]
Oracle & 1.00 & 1.00 & 1.00 & 1.00 & 1.00 & 1.00 & 1.00 & 1.00\\
[\therowspace]
Pool & 1.03 & 1.02 & 1.02 & 1.01 & 1.04 & 1.04 & 1.03 & 1.03\\
[\therowspace]
Small only & 4.13 & 3.34 & 3.31 & 3.18 & 6.04 & 4.43 & 4.55 & 4.78\\
[\therowspace]
Shrink & 4.13 & 2.29 & 2.27 & 2.42 & 6.04 & 2.35 & 2.12 & 1.74\\
[\therowspace]
Weighting & 1.91 & 1.99 & 2.12 & 2.43 & 2.13 & 1.65 & 1.62 & 1.60\\
[\therowspace]
${\widetilde\Theta}_{\mathrm{plug}}$ & 2.17 & 1.73 & 1.69 & 1.56 & 2.80 & 2.01 & 2.00 & 1.95\\
[\therowspace]
${\widetilde\Theta}_{\mathrm{bapi}}$ & 1.77 & 1.39 & 1.33 & 1.19 & 2.13 & 1.52 & 1.47 & 1.31\\
[\therowspace]
${\widehat\Theta}_{\mathrm{bapi}}$ & 1.76 & 1.39 & 1.33 & 1.19 & 2.12 & 1.51 & 1.45 & 1.30\\
[\therowspace]
AICc & 1.73 & 1.35 & 1.30 & 1.15 & 2.06 & 1.47 & 1.41 & 1.24\\
[\therowspace]
\bottomrule
\end{tabular}
\end{table}
 
\section{Inadmissibility}\label{sec:prop}

Section~\ref{sec:simu} gives empirical support
for our proposal. Several of the estimators 
perform better than ordinary shrinkage.
In this section we provide some theoretical support.
We provide a data enriched
estimator that makes least squares on the small
sample inadmissible.  The estimator is derived
for the proportional design case but inadmissibility
holds even when $V_B = X_B^\tran X_B$ is not proportional
to $V_S =X_S^\tran X_S$. The  inadmissibility is
with respect to a loss function 
$\e( \Vert X_T(\hat\beta-\beta) \Vert^2)$
where $V_T = X_T^\tran X_T$ is proportional to $V_S$.

To motivate the estimator, suppose for the moment that
$V_B= N\Sigma$, $V_S=n\Sigma$ and $V_T=\Sigma$
for a positive definite matrix $\Sigma$.
Then the weighting matrix $W_\lambda$ in
Lemma~\ref{lem:matrixweights} simplifies to 
$W_\lambda = \omega I$
where 
$\omega = (N+n\lambda)/(N+n\lambda+N\lambda)$.
As a result $\hat\beta = \omega\hat\beta_S+(1-\omega)\hat\beta_B$
and we can find and estimate an oracle's value for $\omega$.

We show below that the resulting estimator of $\hat\beta$ with
estimated $\omega$ dominates
$\hat\beta_S$ (making it inadmissible) under mild conditions
that do not require $V_B\propto V_S$.
We do need
the model degrees of freedom to be at least $5$,
and it will suffice to have the error degrees of freedom in the small
sample regression be at least $10$.
The result also requires a Gaussian assumption
in order to use a lemma of Stein's. 

Write
$Y_S =  X_{S}\beta+\err_{S}$
and $Y_B  =  X_{B}(\beta+\gamma)+\err_{B}$
for $\err_{S}\simiid \dnorm(0,\sigma_{S}^{2})$ and $\err_{B}\simiid 
\dnorm(0,\sigma_{B}^{2})$.
The mean squared prediction error of 
$\omega\hat\beta_S + (1-\omega)\hat\beta_B$ is 
\begin{align*}
f(\omega) & = \e(\Vert X_T(\hat{\beta}(\omega)-\beta)\Vert^2)\\
 & = \tr((\omega^{2}\sigma^2_S V_{S}^{-1}+(1-\omega)^{2}(\gamma\gamma^{\tran}+
\sigma^2_B V_{B}^{-1}))\Sigma).
\end{align*}
This error is minimized by the oracle's parameter value
\begin{align*}
\omega_\orcl & = \frac{\tr((\gamma\gamma^{\tran}+\sigma^2_B V_{B}^{-1})\Sigma)}
{\tr((\gamma\gamma^{\tran}+\sigma^2_B V_{B}^{-1})\Sigma)+\sigma_{S}^{2}\tr(V_{S}^{-1}\Sigma)}.
\end{align*}

When $V_{S}=n\Sigma$ and $V_{B}=N\Sigma$, we find
\begin{align*}
\omega_{\orcl} & = \frac{\gamma^{\tran}\Sigma\gamma+d\sigma_{B}^{2}/N}
{\gamma^{\tran}\Sigma\gamma+d\sigma_{B}^{2}/N+d\sigma_{S}^{2}/n}.
\end{align*}
The plug-in estimator is
\begin{align}\label{eq:plug-in-weight}
\hat
\omega_{\plug} & = \frac{\hat\gamma^{\tran}\Sigma\hat\gamma+d\hat\sigma_{B}^{2}/N}
{\hat\gamma^{\tran}\Sigma\hat\gamma+d\hat\sigma_{B}^{2}/N+d\hat\sigma_{S}^{2}/n}
\end{align}
where $\hat{\sigma}_{S}^{2}={\Vert Y_S -X_{S}\hat\beta_{S}\Vert^2}/{(n-d)}$
and $\hat{\sigma}_{B}^{2}={\Vert Y_{B}-X_{B}\hat\beta_{B}\Vert^2}/{(N-d)}$.
To allow a later bias adjustment,
we generalize this plug-in estimator.
Let $h(\hat\sigma^2_B)$ be any nonnegative measurable function of
$\hat\sigma^2_B$ with $\e(h(\hat\sigma^2_B))<\infty$.
The generalized plug-in estimator is
\begin{align}\label{eq:gpi}
\hat
\omega_{\plug,h} & = \frac{\hat\gamma^{\tran}\Sigma\hat\gamma+h(\hat\sigma_{B}^{2})}
{\hat\gamma^{\tran}\Sigma\hat\gamma+h(\hat\sigma_{B}^{2})+d\hat\sigma_{S}^{2}/n}.
\end{align}

Here are the conditions under which $\hat\beta_S$ is
made inadmissible by the data enrichment estimator.

\begin{theorem}\label{thm:plug-in}
Let $X_S\in\real^{n\times d}$ and $X_B\in\real^{N\times d}$
be fixed matrices
with $X_{S}^{\tran}X_{S}=n\Sigma$ and $X_{B}^{\tran}X_{B}=V_B$
where $\Sigma$ and $V_B$ both have rank $d$.
Let $Y_S\sim\dnorm(X_S\beta,\sigma^2_SI_n)$
independently of $Y_B\sim\dnorm(X_B(\beta+\gamma),\sigma^2_BI_N)$.
If $d\ge 5$ and $m\equiv n-d\geq10$, then 
\begin{align}
\e(\Vert X_T\hat\beta(\hat\omega)-X_T\beta\Vert^2) 
& < \e(\Vert X_T\hat\beta_{S}-X_T\beta\Vert^2)
\label{eq:theory}
\end{align}
holds for any matrix $X_T$ with $X_T^\tran X_T = \Sigma$
and any $\hat\omega=\hat\omega_{\plug,h}$ given by~\eqref{eq:gpi}.
\end{theorem}
\begin{proof}
 Section~\ref{sec:proofplugin} in the Appendix.
\end{proof}

The condition on $m$ can be relaxed at the
expense of a more complicated statement.
From the details in the proof, it suffices to have
$d\ge5$ and $m(1-4/d)\ge2$.


Because
$\e(\hat{\gamma}^\tran\Sigma\hat{\gamma})
>\gamma^\tran\Sigma\gamma$
we find that $\hat\omega_\plug$ is biased upwards, making it conservative. 
In the proportional design case
we find that the bias is $d\sigma_{S}^{2}/n+d\sigma_{B}^{2}/N$. 
That motivates a bias adjustment, replacing
$\hat\gamma^\tran\Sigma\hat\gamma$
by $\hat\gamma^\tran\Sigma\hat\gamma-d\hat\sigma^2_S/n-d\hat\sigma^2_B/N$.
The result is
\begin{align}\label{eq:bapi-prop}
\hat\omega_{\bapi} 
& = 
\frac{\hat{\gamma}^{\tran}\Sigma\hat\gamma -d\hat{\sigma}_{S}^{2}/n}
{\hat{\gamma}^{\tran}\Sigma\hat{\gamma}}\vee \frac{n}{n+N},
\end{align}
where values below $n/(n+N)$ get rounded up.
This bias-adjusted estimate of $\omega$ is not
covered by Theorem~\ref{thm:plug-in}.
Subtracting only $\hat\sigma^2_B/N$
instead of $\hat\sigma^2_B/N +\hat\sigma^2_S/n$
is covered, yielding
\begin{align}\label{eq:bapigaur}
\hat\omega_{\bapi} '
& = 
\frac {\hat{\gamma}^{\tran}\Sigma\hat{\gamma}}
{\hat{\gamma}^{\tran}\Sigma\hat\gamma +d\hat{\sigma}_{S}^{2}/n},
\end{align}
which corresponds to taking $h(\hat\sigma^2_B)\equiv 0$
in equation~\eqref{eq:gpi}.
Data enrichment with $\hat\omega$ given by~\eqref{eq:bapi-prop}
makes $\hat\beta_S$ inadmissible whether or not the
motivating covariance proportionality holds.

\section{A matrix oracle}\label{sec:moracle}

In this section we look for an explanation
of how data enrichment might be more accurate
than Stein shrinkage.  We generalize our estimator
to 
$$
\hat\beta(W) = W\hat\beta_S + (I-W)\hat\beta_B
$$
and then find the optimal matrix $W$.

\begin{theorem}\label{thm:bestw}
Let $\hat\beta_S\in\real^d$ have mean $\beta$
and covariance matrix $\sigma^2_SV_S^{-1}$ for $\sigma_S>0$.
Let $\hat\beta_B\in\real^d$ be independent of $\hat\beta_S$,
with mean $\beta+\gamma$
and covariance matrix $\sigma^2_BV_B^{-1}$ for $\sigma_B>0$.
Let $\hat\beta(W) = W\hat\beta_S + (I-W)\hat\beta_B$
for a matrix $W\in\real^{d\times d}$.
Let $V_T\in\real^{d\times d}$ be any positive definite
symmetric matrix. Then 
$\e( (\hat\beta(W)-\beta)^\tran V_T (\hat\beta(W)-\beta) )$
is minimized at
\begin{align}\label{eq:bestw}
W= 
(\gamma\gamma^\tran 
+\sigma^2_BV_B^{-1}
+\sigma^2_SV_S^{-1})^{-1}(\gamma\gamma^\tran 
+\sigma^2_BV_B^{-1}).
\end{align}
\end{theorem}
\begin{proof}
Section~\ref{sec:matoracle}.
\end{proof}

It is interesting that when we are free to choose the
entire $d\times d$ matrix $W$, then the optimal choice
is the same for all weighting matrices $V_T$.

The penalized least squares criterion~\eqref{eq:penss}
leads to a matrix weighting of the two within-sample
estimators.  The weighting matrix $W_\lambda$ is in
a one dimensional family indexed by $0\le\lambda\le\infty$.
The optimal $W$ from~\eqref{eq:bestw} is not
generally in that family.

Both $W_\lambda$ from criterion~\eqref{eq:penss} and $W$ from~\eqref{eq:bestw} 
trade off bias and variance, through the appearance
$\gamma\gamma^\tran$, $\sigma^2_S$, and $\sigma^2_B$,
which for~\eqref{eq:penss} appear in
the formula for the optimal $\lambda$.
The advantage of working with~$W_\lambda$
instead of $W$ is that $W_\lambda$
yields a one parameter family of candidate weighting
matrices to search over.

When $V_S$ and $V_B$ are both proportional to the
same positive definite matrix $V_T$, then 
the data enrichment oracle chooses $W=\omega I_d$ where
$$\omega = 
\omega_{\orcl}  = \frac{\tr((\gamma\gamma^{\tran}+\sigma^2_B V_{B}^{-1})V_T)}
{\tr([\gamma\gamma^{\tran}+\sigma^2_B V_{B}^{-1}+\sigma_{S}^{2}V_{S}^{-1}]V_T)}
$$
which mimicks the form of the optimal $W$ in equation~\eqref{eq:bestw},
replacing numerator and denominator by traces after multiplying
both by $V_T$.

The James-Stein shrinker chooses $W=\omega_{\js}I_d$ where
$$
\omega_{\js} = 1-\frac{d-2}{\Vert\hat\theta_S-\hat\theta_B\Vert^2/\sigma^2_S}
= 1-\frac{d-2}{\hat\gamma^\tran V_S\hat\gamma/\sigma^2_S}.
$$
If we approximate $\hat\gamma^\tran V_S\hat\gamma$ by
its expectation
$
\tr( (\gamma\gamma^\tran + \sigma^2_SV_S^{-1}+\sigma^2_BV_B^{-1})V_S)
$ 
we find $\omega_{\js}$ centered around
\begin{align*}
\wt\omega_{\js}=
\frac
{\tr( (\gamma\gamma^\tran +\sigma^2_BV_B^{-1})V_S) + 2\sigma^2_S}
{\tr( (\gamma\gamma^\tran + \sigma^2_BV_B^{-1}+\sigma^2_SV_S^{-1})V_S)}.
\end{align*}
The presence of $2\sigma^2_S$ in the numerator leads
the James-Stein approach to make less aggressive 
use of the big data set than data enrichment does.
We believe that this is why the James-Stein method
did not perform well in our simulations.

\section{Related literatures}\label{sec:lite}

There are many disjoint literatures that study
problems like the one we have presented.
They do not seem to have been compared before,
the literatures seem to be mostly unaware of each other,
and there is a surprisingly large variety of problem contexts.
Some quite similar sounding problems turn out to differ
on critically important details.
We give a brief summary of those topics here.

The key ingredient in our problem is that we care more about
the small sample than the large one. Were that
not the case, we could simply pool all the data
and fit a model with indicator variables picking
out one or indeed many different special subsets of interest.
Without some kind of regularization, 
that approach ends up being similar to taking $\lambda=0$
and hence does not borrow strength.

The closest match to our problem setting comes from small area estimation in survey sampling.
The monograph by \cite{rao:2003} is a comprehensive
treatment of that work and \cite{ghos:rao:1994}
provide a compact summary.
In that context the large sample may be census
data from the entire country and the small sample
(called the small area) may be a single county or
a demographically defined subset.
Every county or demographic group may be taken to be the small sample
in its turn.
The composite estimator \cite[Chapter 4.3]{rao:2003}
is a weighted sum of estimators from small and large
samples. The estimates being combined may be
more complicated than regressions, involving for example
ratio estimates.  The emphasis is usually on scalar
quantities such as small area means or totals, instead of
the regression coefficients we consider.
One particularly useful model \cite[equation (4.2)]{ghos:rao:1994} 
allows the small areas to share regression coefficients
apart from an area specific intercept.
Then BLUP estimation methods lead to shrinkage estimators
similar to ours.

Our methods and results are similar to empirical Bayes
methods, drawing heavily on ideas of Charles Stein.  
A Stein-like result also holds for multiple regression
in the context of just one sample.  
We mentioned already the regression shrinkers
of~\cite{copa:1983} and \cite{stei:1960}.
\cite{efro:morr:1973}
find that the Stein effect for shrinking to a common mean
takes place at dimension $4$ and \cite{geor:1986} finds
that the effect
takes place at dimension $3+q$ when shrinking means towards a $q$--dimensional
linear manifold.

A similar problem to ours is addressed by \cite{chen:chen:2000}.
Like us, they have $(X,Y)$ pairs of both high
and low quality. In their setting
both high and low quality pairs are defined
for the same set of individuals.  Their given sample has
all of the low quality data and the high quality data
are available only on a simple random sample of the subjects.

\cite{boon:mukh:tayl:2013} consider a genomics problem where
there are both low and high quality versions of $X$,
from two different technical platforms, but all data share
the same $Y$. All observations have the low quality $X$'s
while a subset have both high and low quality $X$ measurements.
They take a Bayesian approach.
\cite{boon:tayl:mukh:2012} handle the same problem via
shrinkage estimates.
A crucial difference in our setting, is that the subjects 
are completely different in our two samples; no 
$(X,Y)$ pair in one data set comes from the same person
as an $(X,Y)$ pair in the other data set.

\cite{mukh:chat:2008} use shrinkage methods to blend two
estimators.  One is a case-control estimate of a log odds
ratio. The other is a case-only estimator, derived under an
assumption of gene-environment independence.  They also derive
and employ a plug-in estimator. Their target parameter is scalar
so no Stein effect could be expected.
\cite{chen:chat:carr:2009} address the same issue via
$L_1$ and $L_2$ shrinkage based methods, and 
give some asymptotic covariances.


In chemometrics, a calibration transfer problem
\citep{feud:wood:tan:myle:brow:ferr:2002}
comes up when one wants to adjust a model
to new spectral hardware.  There may be a regression
model linking near-infrared spectroscopy data to
a property of some sample material.  
The transfer problem comes
up for data from a new machine.  Sometimes one
can simply run a selection of samples through both
machines but in other cases that is not possible,
perhaps because one machine is remote
\citep{wood:feud:myle:brow:2004}. Their primary
and secondary instruments correspond to our
small and big samples respectively.
Their emphasis is on transfering
either principal components regression or
partial least squares models, not the plain regressions
we consider here.

A common problem in marketing is data fusion,
also known as statistical matching.
Variables $(X,Y)$ are measured in one sample 
while variables $(X,Z)$ are measured in another.
There may or may not be a third sample with
some measured triples $(X,Y,Z)$.
The goal in data fusion is to use all of the data
to form a large synthetic data set of $(X,Y,Z)$
values, perhaps by imputing missing $Z$ for the
$(X,Y)$ sample and/or missing $Y$ for the $(X,Z)$
sample. When there is no $(X,Y,Z)$ sample, some
untestable assumptions must be made about
the joint distribution, because it cannot be recovered
from its bivariate margins. The text
by \cite{dora:dizi:scan:2006} gives a comprehensive
summary of what can and cannot be done.
Many of the approaches are based on methods
for handling missing data \citep{litt:rubi:2009}.

Medicine and epidemiology among other fields use meta-analysis
\citep{bore:hedg:higg:roth:2009}. 
In that setting there are $(X,Y)$ data sets from numerous
environments, no one of which is necessarily of primary importance.

Our problem is an instance of what machine learning
researchers call domain adaptation.  They may have
fit a model to a large data set (the `source') and then wish to adapt
that model to a smaller specialized data set (the `target').  This
is especially common in natural language processing.
NIPS 2011 included a special session on domain adaptation.
In their motivating problems there are typically a very
large number of features (e.g., one per unique word appearing  in
a set of documents). They also pay special attention
to problems where many of the data points do not
have a measured response. Quite often a computer
can gather high dimensional $X$ while a human 
rater is necessary to produce $Y$.
\cite{daum:2009} surveys various wrapper strategies,
such as fitting a model to weighted combinations
of the data sets, deriving features from the reference
data set to use in the target one and so on.
\cite{cort:mohr:2011} 
consider domain adaptation for kernel-based regularization algorithms, including kernel ridge regression, support vector machines (SVMs), or support vector regression (SVR). They prove pointwise loss guarantees depending on the discrepancy distance between the empirical source and target distributions, and demonstrate the power of the approach on a number of experiments using kernel ridge regression.
We have given conditions under which adaptation is always beneficial.

A related term in machine learning is concept
drift~\citep{widm:kuba:1996}. There a prediction method
may become out of date as time goes on.  The term
drift suggests that slow continual changes are anticipated,
but they also consider that there may be hidden contexts
(latent variables in statistical teminology) affecting some
of the data.

\section{Conclusions}\label{sec:conc}

We have studied a middle ground between pooling a large
data set into a smaller target one and ignoring it completely.
Looking at the left side
of Figures~\ref{fig:locationresults},
\ref{fig:regresults0}  and~\ref{fig:regresults1} 
we see that in the low bias cases
the more aggressive methods have a clear advantage. Fortune
favors the bold. Pooling is the boldest and wins the most
when bias is small.
But pooling has unbounded risk as bias increases.
That is, misfortune also favors the bold.
Our shrinkage methods provide a compromise.  In higher dimensional settings 
of Figures~\ref{fig:regresults0}  and~\ref{fig:regresults1}, we see 
that AICc and bias adjusted plug-in gain a lot of efficiency
when the bias is low. When the bias is high, they are squeezed into
a narrow band between the oracle performance and that of $\hat\beta_S$
which ignores the big data set.
As a result, the new methods show large improvements compared
to shrinkage when the bias small but only lose a little
when the bias is large.


\section*{Acknowledgments}

We thank the following people for helpful discussions:
Penny Chu,   Corinna Cortes, Tony Fagan, Yijia Feng, Jerome Friedman, Jim Koehler,
Diane Lambert, Elissa Lee  and Nicolas Remy.  

\bibliographystyle{apalike}
\bibliography{data-enrich}

\section{Appendix I: proofs}\label{sec:proofs}

This appendix presents proofs of the results in
this article. They are grouped into sections by topic,
with some technical supporting lemmas separated
into their own sections.

\subsection{Proof of Theorem~\ref{thm:df}}\label{sec:thmdf}

\begin{proof}
First
$$\df(\lambda)
= \sigma^{-2}_S\tr( \cov(X_S\hat\beta,Y_S))
= \sigma^{-2}_S\tr( X_SW_\lambda (X_S^\tran X_S)^{-1}X_S^\tran\sigma^2_S)
=\tr(W_\lambda).$$
Next with $X_T=X_S$, and $M=V_S^{1/2}V_B^{-1}V_S^{1/2}$,
\begin{align*}
\tr(W_\lambda) = \tr(V_S+\lambda V_SV_B^{-1}V_S+\lambda V_S)^{-1}(V_S+\lambda V_SV_B^{-1}V_S).
\end{align*}
We place $V_S^{1/2}V_S^{-1/2}$ between these factors
and absorb them left and right.  Then we reverse
the order of the factors and repeat the process, yielding
\begin{align*}
\tr(W_\lambda) = \tr(I+\lambda M+\lambda I)^{-1}(I+\lambda M).
\end{align*}
Writing $M = U\diag(\nu_1,\dots,\nu_d)U^\tran$ for 
an orthogonal matrix $U$ and simplifying yields the result.
\end{proof}

\subsection{Proof of Theorem~\ref{thm:pmse}}\label{sec:pmse}

\begin{proof}
First 
$\e( \Vert X_{T}\hat{\beta}-X_{T}\beta\Vert^{2})  =\tr(V_{S}\e ((\hat{\beta}-\beta)(\hat{\beta}-\beta)^\tran))$.
Next using $W=W_\lambda$, we make a bias-variance decomposition,
\begin{align*}
\e\bigl( (\hat{\beta}-\beta)(\hat{\beta}-\beta)^\tran \bigr)
& =(I-W)
\gamma\gamma^\tran(I-W)^\tran+\cov(W\hat{\beta}_{S})+\cov((I-W)\hat{\beta}_{B})\\
 & =\sigma_{S}^{2}WV_{S}^{-1}W^\tran+(I-W)\Theta(I-W)^\tran,
\end{align*}
for $\Theta = \gamma\gamma^\tran + \sigma^2_BV_B^{-1}$.
Therefore
$\e\bigl( \Vert X_{S}(\hat{\beta}-\beta)\Vert^{2}\bigr) 
=\sigma_S^{2}\tr(V_SWV_S^{-1}W^\tran)+\tr(\Theta(I-W)^\tran V_S(I-W)).$

Now we introduce $\wt W=V_S^{1/2}WV_S^{-1/2}$ finding
\begin{align*}
\wt W & =V_S^{1/2}(V_B+\lambda V_S + \lambda V_B)^{-1}(V_B+\lambda V_S)V_S^{-1/2}\\
 & =(I+\lambda M+\lambda I)^{-1}(I + \lambda M)\\
 & = U \wt D U^\tran,
\end{align*}
where $\wt D = \diag((1+\lambda\nu_j)/(1+\lambda +\lambda \nu_j))$.
This allows us to write the first term of the mean squared error as
$$
\sigma^2_S\tr(V_SWV_S^{-1}W^\tran)
=\sigma^2_S\tr(\wt W\wt W^\tran)
=\sigma^2_S\sum_{j=1}^{d}\frac{(1+\lambda\nu_{j})^{2}}{(1+\lambda+\lambda v_{j})^{2}}.
$$
For the second term,
let $\wt\Theta=V_S^{1/2}\Theta V_S^{1/2}$. Then 
\begin{align*}
\tr\bigl( \Theta(I-W)^\tran V_S(I-W)\bigr)
 & =\tr(\wt\Theta(I-\wt W)^\tran(I-\wt W))\\
 & =\tr(\tilde{\Theta}U(I-\wt D)^{2}U^\tran) \\
 &
 =\lambda^2\sum_{k=1}^{d}\frac{
u_{k}^\tran V_{S}^{1/2}\Theta
V_{S}^{1/2}u_{k}}{(1+\lambda+\lambda\nu_{k})^{2}}.\ \qedhere
\end{align*}
\end{proof}

\subsection{Derivation of equation~\eqref{eq:cvk}}\label{sec:cvderive}

We suppose for simplicity that $n = rK$ for an integer
$r$, so the $K$ folds have equal size.
In that case $\bar Y_{S,-k} = (n\bar Y_S -r\bar Y_{S,k})/(n-r)$.
Now
\begin{align}\label{eq:omegacv}
\omega_\cv 
&=
\frac
{\sum_k(\bar Y_{S,-k}-\bar Y_B)(\bar Y_{S,k}-\bar Y_B)}
{\sum_k (\bar Y_{S,-k}-\bar Y_B)^2}
\end{align}
After some algebra, the numerator of~\eqref{eq:omegacv}
is
$$K(\bar Y_S-\bar Y_B)^2 -\frac{r}{n-r}\sum_{k=1}^K(\bar Y_{S,k}-\bar Y_S)^2$$
and the denominator is
$$K(\bar Y_S-\bar Y_B)^2 +
\biggl(\frac{r}{n-r}\biggr)^2\sum_{k=1}^K(\bar Y_{S,k}-\bar Y_S)^2.$$
Letting  $\hat\delta_0 = \bar Y_B-\bar Y_S$ and
$\hat\sigma^2_{S,K} = (1/K)\sum_{k=1}^K(\bar Y_{S,k}-\bar Y_S)^2$,
we have
$$
\omega_\cv 
= 
\frac
{ \hat\delta_0^2 - \hat\sigma^2_{S,K}/(K-1)}
{ \hat\delta_0^2 + \hat\sigma^2_{S,K}/(K-1)^2}.
$$

The only quantity in $\omega_\cv$ which depends on
the specific $K$-way partition used is
$\hat\sigma^2_{S,K}$.
If the groupings are chosen by sampling without
replacement, then under this sampling,
$$
\e(\hat\sigma^2_{S,K})=
\e( (\bar Y_{S,1}-\bar Y_S)^2)
=
\frac{s^2_S}r(1-1/K)
$$
using the finite population correction 
for simple random sampling,
where $s^2_S = \hat\sigma^2_S n/(n-1)$.
This simplifies to 
$$\e(\hat\sigma^2_{S,K})
=\hat\sigma^2_S\frac{n}{n-1}\frac1r\frac{K-1}K
=\hat\sigma^2_S\frac{K-1}{n-1}.
$$
Replacing $\hat\sigma^2_{S,K}$ in $\omega_{\cv}$ by its expectation
yields~\eqref{eq:cvk}.

\subsection{Proof of Theorem~\ref{thm:l1closed}}\label{sec:proofl1closed}
\begin{proof}
If $\lambda >nN|\bar Y_B-\bar Y_S|/(n+N)$ then we may
find directly that with any value of $\delta>0$
and corresponding $\mu$ given by~\eqref{eq:l1locest},
the derivative of~\eqref{eq:l1loccrit} with respect to $\delta$
is positive. Therefore $\hat\delta\le 0$ and a similar argument
gives $\hat\delta\ge0$, so that $\hat\delta=0$ and then
$\hat\mu = (n\bar Y_S+N\bar Y_B)/(n+N)$.

Now suppose that $\lambda\le \lambda_*$.
We verify that the quantities in~\eqref{eq:l1soln} 
jointly satisfy equations~\eqref{eq:l1locest}.
Substituting $\hat\delta$ from~\eqref{eq:l1soln} into
the first line of~\eqref{eq:l1locest} yields
\begin{align*}
\frac{n\bar Y_S+N(\bar Y_S+\lambda(N+n)\eta/(Nn))}{n+N}
= \bar Y_S + \frac\lambda{n}\sign(\bar Y_B-\bar Y_S),
\end{align*}
matching the value in~\eqref{eq:l1soln}.
Conversely, substituting $\hat\mu$ from~\eqref{eq:l1soln} into
the second line of~\eqref{eq:l1locest} yields
\begin{align}\label{eq:threshed}
\Theta\Bigl(\bar Y_B-\hat\mu;\frac\lambda{N}\Bigr)
&=\Theta\Bigl(\bar Y_B-\bar Y_S -\frac\lambda{n}\sign(\bar Y_B-\bar Y_S);\frac\lambda{N}\Bigr).
\end{align}
Because of the upper bound on $\lambda$, the result is
$\bar Y_B-\bar Y_S-\lambda(1/n+1/N)\sign(\bar Y_B-\bar Y_S)$
which matches the value in~\eqref{eq:l1soln}.
\end{proof}

\subsection{Derivation of equation~\eqref{eq:l1loss}}\label{sec:proofl1orcl} 

Let $f=n/(n+N)$ be the fraction of the pooled data
coming from the small sample and $F=1-f$ be the
fraction from the large sample.
Define $D=\bar Y_B-\bar Y_S$ and $c=\lambda/(nF)=\lambda(1/n+1/N)$.
Then
\begin{align*}
\hat\mu & = 
\begin{cases}
\bar Y_S + {\lambda}/n\, \sign(D), & |D| \ge c\\
f\bar Y_S + F\bar Y_B, & |D| \le c
\end{cases}\\
& = F\bar Y_B+f\bar Y_S
+((\lambda/n)\sign(D)-FD)\onedc.
\end{align*}

We replace $\bar Y_B$ and $\bar Y_S$ by linear
combinations of $D$ and another variable
$H$ chosen to be statistically independent of $D$.
Specifically,
$H=\bar Y_B + \alpha \bar Y_S$
for $\alpha = (\sigma^2_B/N)/(\sigma^2_S/n)$.
The inverse transformation is
$$
\begin{pmatrix}
\bar Y_S \\
\bar Y_B
\end{pmatrix}
=
\frac1{\alpha+1}
\begin{pmatrix}
-1 & 1\\
\alpha & 1\\
\end{pmatrix}
\begin{pmatrix}
D\\
H
\end{pmatrix}.
$$
In terms of these independent variables we have
$$
\hat\mu = 
c_0D + c_1 H + ((\lambda/n)\sign(D)-FD)\onedc
$$
where $c_0 = \alpha/(\alpha+1)-f$, and $c_1 = 1/(\alpha+1)$.

Without loss of generality, $\mu=0$.
Then $D\sim\dnorm(\delta,\sigma^2_S/n+\sigma^2_B/N)$
independently of $H\sim\dnorm(\delta,\alpha^2\sigma^2_S/n + \sigma^2_B/N)$.
After some algebra,
\begin{align}\label{eq:postalgebra}
\begin{split}
\e(\hat\mu-\mu)^2
& =
c_0^2 \e(D^2)
+c_1^2\e(H^2)
+2c_0c_1\e(D)\e(H)\\
&\quad +(\lambda/n)^2\e(\onedc) - 2c_1F\e(H)\e\bigl(D\onedc\bigr)\\
&\quad + F(F-2c_0)\e\bigl(D^2\onedc\bigr)\\
&\quad + 2c_1(\lambda/n)\e(H)\e\bigl(\sign(D)\onedc\bigr)\\
&\quad + 2(\lambda/n)(c_0-F)\e\bigl(D\,\sign(D)\onedc\bigr).
\end{split}
\end{align}

In addition to first and second moments of $D$ and $H$, we
need some expectations of functions of $D$ involving the
sign function and some indicators. They are given by
Lemma~\ref{lem:dparts} below.
Let $\varphi$ and $\Phi$ be the probabilty density
and cumulative distribution functions respectively of
the $\dnorm(0,1)$ distribution.

\begin{lemma}\label{lem:dparts}
Let $D\sim\dnorm(\delta,\tau^2)$.
Define $\eta_+ = (\delta-c)/\tau$, $\eta_- = (-\delta-c)/\tau$,
$\Phi_\pm = \Phi(\eta_\pm)$ and $\varphi_\pm = \varphi(\eta_\pm)$.
Then for $c\ge0$,
\begin{align}
\e( 1_{|D|\ge c})& = \Phi_+ +\Phi_-\label{eq:justdc}\\
\e( D1_{|D|\ge c} )
& = 
\delta(\Phi_++\Phi_-) + \tau(\varphi_+-\varphi_-) \label{eq:dc}\\
\e( D^21_{|D|\ge c} )& = 
(\delta^2+\tau^2)(\Phi_++\Phi_-)\label{eq:ddc}\\
&\quad+\tau c(\varphi_++\varphi_-)
+\tau \delta(\varphi_+-\varphi_-)\notag\\
\e( \sign(D)1_{|D|\ge c} ) & = \Phi_+-\Phi_-,\quad\text{and}
\label{eq:sdc}\\
\e( D\,\sign(D)1_{|D|\ge c} ) & = 
\delta(\Phi_+ -\Phi_-)
+\tau( \varphi_+ +\varphi_-).
\label{eq:dsdc}
\end{align}
\end{lemma}
\begin{proof}
Equation~\eqref{eq:justdc} is almost immediate.
For $Z\sim\dnorm(0,1)$, using
Chapter 2.5.1 of \cite{pate:read:1996} yields
$\e( Z1_{Z\le c}) = g_1(c)\equiv -\varphi(c)$
and $\e( Z^21_{Z\le c})= g_2(c)\equiv \Phi(c)-c\varphi(c)$.
For $D\sim\dnorm(\delta,\tau^2)$  we may write
$D = \delta + \tau Z$ and then
\begin{align*}
\e( D1_{D\le c}) & = g_1(c,\delta,\tau)\equiv
\delta\Phi\Bigl(\frac{c-\delta}\tau\Bigr) + \tau g_1\Bigl(\frac{c-\delta}\tau\Bigr)\\
& = \delta\Phi\Bigl(\frac{c-\delta}\tau\Bigr)
-\tau\varphi\Bigl(\frac{c-\delta}\tau\Bigr)
,\quad\text{and}\\
\e(D^21_{D\le c}) & = g_2(c,\delta,\tau)\equiv 
\delta^2\Phi\Bigl(\frac{c-\delta}\tau\Bigr) +2\delta\tau g_1\Bigl(\frac{c-\delta}\tau\Bigr)
+\tau^2g_2\Bigl(\frac{c-\delta}\tau\Bigr)\\
& = (\delta^2+\tau^2)\Phi\Bigl(\frac{c-\delta}\tau\Bigr)
-\tau (c+\delta)\varphi\Bigl(\frac{c-\delta}\tau\Bigr).
\end{align*}
For $c>0$, we write
$1_{|D|\ge c} = 1_{D\le -c} + 1_{D\ge c}=1_{D\le -c} + 1_{-D\le -c}$,
and so
$\e( D1_{|D|\ge c}) = g_1(-c,\delta,\tau) -g_1(-c,-\delta,\tau)$
which simplifies to~\eqref{eq:dc}.
Similarly
$\e(D^21_{|D|\ge c}) = g_2(-c,\delta,\tau) +g_2(-c,-\delta,\tau)$
which simplifies to~\eqref{eq:ddc}.

Equation~\eqref{eq:sdc} 
follows upon writing
$1_{|D|\ge c} = 1_{D\ge c}-1_{-D\ge c}$.
For~\eqref{eq:dsdc} that step yields
$\e(D\,\sign(D)1_{|D|\ge c})=g_1(-c,-\delta,\tau)-g_1(-c,\delta,\tau)$
\end{proof}

The formula in the article follows by making substitutions of the
quantities from Lemma~\ref{lem:dparts}
into equation~\eqref{eq:postalgebra}. It also uses the identity $c_0+c_1=F$.

\subsection{Supporting lemmas for inadmissibility}\label{sec:lemmas}

In this section we first recall Stein's Lemma.  Then we prove
two technical lemmas used in the proof of Theorem~\ref{thm:plug-in}.

\begin{lemma}\label{lem:stein}
Let $Z\sim\dnorm(0,1)$ and let $g:\real\to\real$
be an indefinite integral of the Lebesgue measurable function
$g'$, essentially the derivative of $g$. If $\e(|g'(Z)|)<\infty$
then
$$
\e( g'(Z) ) = \e( Zg(Z)).
$$
\end{lemma}
\begin{proof}
\cite{stei:1981}.
\end{proof}

\begin{lemma}\label{lem:pmse}
Let $\eta\sim\dnorm(0,I_{d})$, $b\in \real^{d}$,
and let $A>0$ and $B>0$ be constants. 
Let 
$$ Z = \eta+\frac{A(b-\eta)}{\Vert b-\eta\Vert^2+B}.$$
Then
\begin{align*}
\e(\Vert Z\Vert^2) 
& <
d+\e\left(\frac{A(A+4-2d)}{\Vert b-\eta\Vert^2+B}\right).
\end{align*}
\end{lemma}
\begin{proof}
First,
\begin{align*}
\e(\Vert Z\Vert^2)
& = d
+\e\left(\frac{A^{2}\Vert b-\eta\Vert^2}{(\Vert b-\eta\Vert^2+B)^{2}}\right)
+2A\sum_{k=1}^{d}\e\left(\frac{\eta_{k}(b_{k}-\eta_{k})}{\Vert b-\eta\Vert^2+B}\right).
\end{align*}
Now define
$$g(\eta_{k})
=\frac{b_{k}-\eta_{k}}{\Vert b-\eta\Vert^2+B}
=\frac{b_{k}-\eta_{k}}{(b_k-\eta_k)^2+\Vert b_{-k}-\eta_{-k}\Vert^2+B}.
$$ 
By Stein's lemma (Lemma~\ref{lem:stein}), we have
\begin{align*}
\e\left(\frac{\eta_{k}(b_{k}-\eta_{k})}{\Vert b-\eta\Vert^2+B}\right) & = 
\e(g'(\eta_k))=
\e\left(\frac{2(b_{k}-\eta_{k})^{2}}{(\Vert b-\eta\Vert^2+B)^{2}}-\frac{1}{\Vert b-\eta\Vert^2+B}\right)
\end{align*}
and thus
\begin{align*}
\e( \Vert Z\Vert^2)
& =d
+\e\left(\frac{(4A+A^{2})\Vert b-\eta\Vert^2}{(\Vert b-\eta\Vert^2+B)^{2}}
        -\frac{2Ad}{\Vert b-\eta\Vert^2+B}\right)\\
& =d
+\e\left(\frac{(4A+A^{2}-2Ad) }{\Vert b-\eta\Vert^2+B}
        -\frac{(4A+A^2)B}{(\Vert b-\eta\Vert^2+B)^2}\right),
\end{align*}
after collecting terms.
\end{proof}

\begin{lemma}\label{lem:chisq}
For integer $m\ge 1$, let $Q\sim\chi_{(m)}^{2}$, $C>1$, $D>0$
and put $$Z=\frac{Q(C-m^{-1}Q)}{Q+D}.$$
Then
\begin{align*}
\e(Z)
& \geq 
\frac{(C-1)m-2}{m+2+D}.
\end{align*}
and so $\e(Z)>0$ whenever $C>1+{2}/{m}$.
\end{lemma}
\begin{proof}
The $\chi^2_{(m)}$ density function is
$p_{m}(x)= ({2^{m/2-1}\Gamma(\frac{m}{2})})^{-1}x^{m/2-1}e^{-x/2}$.
Thus
\begin{align*}
\e(Z)
& = \frac{1}{2^{m/2}\Gamma(\frac{m}{2})}\int_{0}^{\infty}\frac{x(C-m^{-1}x)}{x+D}x^{m/2-1}e^{-x/2}\rd x\\
 & = m\int_{0}^{\infty}\frac{C-m^{-1}x}{x+D}p_{m+2}(x)\rd x\\
 & \ge
m\frac{C-(m+2)/m}{m+2+D}
\end{align*}
by Jensen's inequality.
\end{proof}

\subsection{Proof of Theorem~\ref{thm:plug-in}}\label{sec:proofplugin}

We prove this first for $\hat\omega_{\plug,h}=\hat\omega_{\plug}$,
that is, taking $h(\hat\sigma^2_B)=d\hat\sigma^2_B/n$.
We also assume at first that $V_B=N\Sigma$ but
remove the assumption later.

Note that $\hat\beta_{S}=\beta+(X_{S}^{\tran}X_{S})^{-1}X_{S}^{\tran}\err_{S}$
and $\hat\beta_{B}=\beta+\gamma+(X_{B}^{\tran}X_{B})^{-1}X_{B}^{\tran}\err_{B}$.
It is convenient to define
\begin{align*}
\eta_{S} & = \Sigma^{1/2}(X_{S}^{\tran}X_{S})^{-1}X_{S}^{\tran}\err_{S}
\quad\text{and}\quad
\eta_{B}  = \Sigma^{1/2}(X_{B}^{\tran}X_{B})^{-1}X_{B}^{\tran}\err_{B}.
\end{align*}
Then we can rewrite $\hat\beta_{S}=\beta+\Sigma^{-1/2}\eta_{S}$
and $\hat\beta_{B}=\beta+\gamma+\Sigma^{-1/2}\eta_{B}$. 
Similarly, we let
\begin{align*}
\hat{\sigma}_{S}^{2} & = \frac{\Vert Y_S-X_S\hat\beta_S\Vert^2}{n-d}
\quad\text{and}\quad
\hat{\sigma}_{B}^{2} = \frac{\Vert Y_B-X_B\hat\beta_B\Vert^2}{N-d}.
\end{align*}
Now
$(\eta_{S},\eta_{B},\hat{\sigma}_{S}^{2},\hat{\sigma}_{B}^{2})$ are
mutually independent, with
\begin{align*}
\eta_{S} & \sim \mathcal{N}\Bigl(0, \frac{\sigma_{S}^{2}}nI_{d}\Bigr),
& \eta_{B} & \sim \mathcal{N}\Bigl(0,\frac{\sigma_B^2}NI_{d}\Bigr),\\
\hat{\sigma}_{S}^{2} & \sim \frac{\sigma_{S}^{2}}{n-d}\chi_{(n-d)}^{2},\quad\text{and}
&\hat{\sigma}_{B}^{2} & \sim \frac{\sigma_{B}^{2}}{N-d}\chi_{(N-d)}^{2}.
\end{align*}

We easily find that $\e( \Vert X\hat\beta_S-X\beta\Vert^2) = d\sigma^2_S/n$.
Next we find $\hat\omega$ and a bound on 
$\e( \Vert X\hat\beta(\hat\omega)-X\beta\Vert^2) $.

Let $\gamma^{*}=\Sigma^{1/2}\gamma$ 
so that $\hat{\gamma}=\hat\beta_B-\hat\beta_S=
\Sigma^{-1/2}(\gamma^{*}+\eta_{B}-\eta_{S})$. Then
\begin{align*}
\hat\omega = \hat{\omega}_\plug & =
\frac{\hat{\gamma}^{\tran}\Sigma\hat{\gamma}+d\hat{\sigma}_{B}^{2}/N}
{\hat{\gamma}^{\tran}\Sigma\hat{\gamma}+d\hat{\sigma}_{B}^{2}/N
+d\hat{\sigma}_{S}^{2}/n}\\
&= \frac{\Vert\gamma^{*}+\eta_{B}-\eta_{S}\Vert^2+d\hat{\sigma}_{B}^{2}/N}
{\Vert\gamma^{*}+\eta_{B}-\eta_{S}\Vert^2+d(\hat{\sigma}_{B}^{2}/N+\hat{\sigma}_{S}^{2}/n)}.
\end{align*}
Now we can express the mean squared error as
\begin{align*}
 \e(\Vert X\hat\beta(\hat\omega)-X\beta\Vert^2)
 & = \e(\Vert X\Sigma^{-1/2}(\hat{\omega}\eta_{S}+(1-\hat{\omega})(\gamma^*+\eta_{B}))\Vert^2)\\
 & = \e(\Vert\hat\omega\eta_{S}+(1-\hat{\omega})(\gamma^{*}+\eta_{B})\Vert^2)\\
 & = \e(\Vert\eta_{S}+(1-\hat{\omega})(\gamma^{*}+\eta_{B}-\eta_{S})\Vert^2)\\
 & = \e\biggl(\Bigl\Vert\eta_{S}+\frac{
(\gamma^{*}+\eta_{B}-\eta_{S})d\hat{\sigma}_{S}^{2}/n}
{\Vert\gamma^{*}+\eta_{B}-\eta_{S}\Vert^2+d(\hat{\sigma}_{B}^{2}/N+\hat{\sigma}_{S}^{2}/n)}
\Bigr\Vert^2\biggr).
\end{align*}

To simplify the expression for mean squared error we  introduce
\begin{align*}
Q &= m\hat\sigma^2_S/\sigma^2_S\sim\chi^2_{(m)}\\
\eta_{S}^{*} &=\sqrt{n}\,\eta_{S}/\sigma_S\sim\mathcal{N}(0,I_{d}),\\
b&=\sqrt{n}
(\gamma^{*}+\eta_{B})/\sigma_{S},\\
A&=d\hat{\sigma}_{S}^{2}/\sigma_{S}^{2}=dQ/m,\quad\text{and}\\
B&=nd(\hat{\sigma}_{B}^{2}/N+\hat{\sigma}_{S}^{2}/n)/\sigma_{S}^{2}\\
&=d( (n/N)\hat{\sigma}_{B}^{2}/\sigma_{S}^{2}+Q/m).
\end{align*}

The quantities $A$ and $B$ are, after conditioning, the
constants that appear in technical Lemma~\ref{lem:pmse}.
Similarly $C$ and $D$ introduced below match
the constants used in Lemma~\ref{lem:chisq}.

With these substitutions and some algebra,
\begin{align*}
 \e(\Vert X\hat\beta(\hat\omega)-X\beta\Vert^2)
 & = 
\frac{\sigma_{S}^{2}}n\,
\e\left(
\left\Vert\eta_{S}^{*}+\frac{A(b-\eta_{S}^{*})}{\Vert b-\eta_{S}^{*}\Vert^2+B}\right\Vert^2\right)\\
 & = \frac{\sigma_{S}^{2}}n\,
\e\left(\e\left(
\left\Vert\eta_{S}^{*}+\frac{A(b-\eta_{S}^{*})}{\Vert b-\eta_{S}^{*}\Vert^2+B}\right\Vert^2
\Bigm|
\eta_{B},\hat{\sigma}_{S}^{2},\hat{\sigma}_{B}^{2}\right)\right).
\end{align*}
We now apply the two technical lemmas from Section~\ref{sec:lemmas}.

Since $\eta_{S}^{*}$ is independent of $(b,A,B)$ and $Q\sim\chi_{(m)}^{2}$,
by Lemma \ref{lem:pmse}, we have
\begin{align*}
  \e\left(\left\Vert\eta_{S}^{*}+\frac{A(b-\eta_{S}^{*})}{\Vert b-\eta_{S}^{*}\Vert^2+B}\right\Vert^2
\Bigm|\eta_{B},\hat{\sigma}_S^{2},\hat\sigma^2_B\right) 
&<d+\e\left(\frac{A(A+4-2d)}{\Vert b-\eta^*_S\Vert^2+B}\Bigm| \eta_B,\hat\sigma_S^2,\hat\sigma_B^2\right).
\end{align*}
Hence
\begin{align}
\Delta&\equiv\e(\Vert X\hat\beta_{S}-X\beta\Vert^2)
-\e(\Vert X\hat\beta(\hat\omega)-X\beta\Vert^2)\notag \\
&>\frac{\sigma^2_S}{n}
\e\left(\frac{A(2d-A-4)}{\Vert b-\eta^*_S\Vert^2+B}\right)\notag\\
&=\frac{d\sigma^2_S}{n}
\e\left(\frac{Q(2-Q/m-4/d)}{\Vert b-\eta^*_S\Vert^2m/d+(B-A)m/d+Q}\right)\notag\\
&=\frac{d\sigma^2_S}{n}\e\left(\frac{Q(C-Q/m)}{Q+D}\right)\label{eq:chisqquantity}
\end{align}
where 
$C=2-4/d$ and
$D=(m/d)(\Vert b-\eta^*_S\Vert^2+dnN^{-1}\hat\sigma^2_B/\sigma^2_S)$.

Now suppose that $d\ge 5$.
Then $C\ge 2-4/5>1$
and so conditionally on $\eta_S$, $\eta_B$, and $\hat\sigma^2_B$,
the requirements of Lemma~\ref{lem:chisq}
are satisfied by $C$, $D$ and $Q$.
Therefore
\begin{align}\label{eq:almostthere}
\Delta &\ge 
\frac{d\sigma^2_S}{n} \e\left(
\frac{m(1-4/d) -2}{m+2+D}\right)
\end{align}
where the randomness in~\eqref{eq:almostthere} is
only through $D$ which depends on $\eta^*_S$, $\eta_B$ (through $b$) and 
$\hat\sigma^2_B$.
By Jensen's inequality
\begin{align}\label{eq:jenn}
\Delta & >
\frac{d\sigma^2_S}{n}
\frac{m(1-4/d) -2}{m+2+\e(D)} \ge 0
\end{align}
whenever $m(1-4/d) \ge 2$.
The first inequality  in~\eqref{eq:jenn}
is strict because $\var(D)>0$.
Therefore $\Delta>0$.
The condition on $m$ and $d$ holds for any $m\ge 10$ when $d\ge 5$.

For the general plug-in $\hat\omega_{\plug,h}$ we replace
$d\hat\sigma^2_B/N$ above by $h(\hat\sigma^2_B)$.
This quantity depends on $\hat\sigma^2_B$ and is
independent of $\hat\sigma^2_S$, $\eta_B$ and $\eta_S$.
It appears within $B$ where we need it to be non-negative
in order to apply Lemma~\ref{lem:pmse}. It also appears
within $D$ which becomes
$(m/d)(\Vert b-\eta^*_S\Vert^2+nh(\hat\sigma^2_B)/\sigma^2_S)$.
Even when we take $\var(h(\hat\sigma^2_B))=0$
we still get $\var(D)>0$ and so
the first inequality in~\eqref{eq:jenn}
is still strict.  

Now suppose that $V_B\ne N\Sigma$.
The distributions of $\eta_S$, $\hat\sigma^2_S$ and
$\hat\sigma^2_B$ remain unchanged but now
$$\eta_B\sim\dnorm\Bigl(0, \Sigma^{1/2}V_B^{-1}\Sigma^{1/2}\sigma^2_B\Bigr)$$
independently of the others. The changed distribution of $\eta_B$ does
not affect the application of Lemma~\ref{lem:pmse} because that
lemma is invoked conditionally on $\eta_B$. Similarly,
Lemma~\ref{lem:chisq} is applied conditionally on $\eta_B$.
The changed distribution of $\eta_B$ changes the distribution of $D$
but we can still apply~\eqref{eq:jenn}.
$\Box$
\bigskip

The expectation at~\eqref{eq:chisqquantity} is negative
when $d=4$, 
as can be verified by a one dimensional quadrature.
For this reason, the inadmissibilty result requires
$d>4$.


\subsection{Proof of Theorem~\ref{thm:bestw}}\label{sec:matoracle}

Recall that $\hat\beta(W) = W\hat\beta_S + (I-W)\hat\beta_B$.
The first two  moments of $\hat\beta(W)$ are
$\e(\hat\beta(W))=\beta + (I-W)\gamma$, and 
$$\var(\hat\beta(W))=\sigma^2_SWV_SW^\tran
+\sigma^2_B(I-W)V_B(I-W)^\tran.$$
The loss is $L(W) = \e( (\hat\beta(W)-\beta)^\tran V_T (\hat\beta(W)-\beta) )$
and
\begin{equation*}
\begin{split}
L(W)
&= 
\tr(\gamma\gamma^\tran V_T) +\tr( WV_TW^\tran\gamma\gamma^\tran)
 -2\tr(  WV_T\gamma\gamma^\tran  )
\\
&\phantom{=}+ \sigma^2_S\tr(WV_SW^\tran V_T)
+ \sigma^2_B\tr(V_B)
+ \sigma^2_B\tr( WV_BW^\tran V_T )
-2\sigma^2_B\tr(WV_BV_T).
\end{split}
\end{equation*}

We will use two
rules from \cite{broo:2011} for matrix differentials.
If $A$, $B$ and $C$ don't depend on the matrix $X$
then the differential of $\tr(XA)$ and $\tr(AX)$ are both $AdX$,
and, the differential
of $\tr(AXBX^\tran C)$ is $BX^\tran CA+B^\tran X^\tran A^\tran C^\tran$
times $dX$.

The differential of $L(W)$ when $W$ changes is
\begin{align*}
&V_TW^\tran \gamma\gamma^\tran+V_T^\tran W^\tran \gamma\gamma^\tran
-2V_T\gamma\gamma^\tran\\
&+\sigma^2_S
\bigl(V_SW^\tran V_T+V_S^\tran W^\tran V_T^\tran\bigr)\\
&+\sigma^2_B
\bigl(V_BW^\tran V_T+V_B^\tran W^\tran V_T^\tran\bigr)
-2\sigma^2_B V_BV_T
\end{align*}
times $dW$.
Let $W^*$ be the hypothesized optimal matrix given
at~\eqref{eq:bestw}. It is symmetric, as are $V_S$, $V_B$
and $V_T$. We may therefore write the differential at
that matrix as
$$
2(GW^* -G+\sigma_S^2V_SW^* + \sigma_B^2V_BW^*-\sigma^2_BV_B)V_T
$$
where $G=\gamma\gamma^\tran$. 
This differential vanishes, showing that $W^*$ satisfyies
first order conditions.
The differential of this differential is
$2(G+\sigma_S^2V_S + \sigma_B^2V_BW)V_T$
which is positive definite, and so $W^*$ must be a minimum.
$\Box$


\end{document}